\documentclass[11 pt]{article}
\usepackage{amsfonts,amsmath,color,amsthm,amssymb, url, enumerate, bbm, subfig}
\usepackage{graphicx}
\usepackage[DIV=14,BCOR=2mm,headinclude=true,footinclude=false]{typearea}
\usepackage[font=small,labelfont=bf]{caption}
\usepackage{hyperref}
\usepackage{tikz, etoolbox}
\usetikzlibrary{shapes}
\usetikzlibrary{arrows}
\usepackage[margin=1in]{geometry}

\newcommand{\ch}{\choose}
\renewcommand{\ch}[2]{{#1 \choose #2}}

\newcommand{\eps}{\epsilon}

\newcommand{\fa}{\text{ for all }}

\newcommand{\al}{\alpha}

\newcommand{\f}{\frac}

\newcommand{\tr}{\text{trace}}

\renewcommand{\th}{\theta}

\newcommand{\N}{\mathbb{N}}
\newcommand{\R}{\mathbb{R}}

\definecolor{purple}{RGB}{128,0,128}

\def\qed{\vbox{\hrule\hbox{\vrule\kern3pt\vbox{\kern6pt}\kern3pt\vrule}\hrule}}

\def\eps{\epsilon}

\newcommand{\SDPOPT}{\text{OPT}_{\text{SDP}}}
\newcommand{\TSPOPT}{\text{OPT}_{\text{TSP}}}
\newcommand{\CYCOPT}{\text{OPT}_{\text{k-Cycle}}}

\newtheorem{thm}{Theorem}[section]
\newtheorem{lm}[thm]{Lemma}
\newtheorem{cor}[thm]{Corollary}
\newtheorem{rem}[thm]{Remark}
\newtheorem{prop}[thm]{Proposition}

\newtheorem{cm}[thm]{Claim}

\newtheorem*{thm*}{Theorem}

\theoremstyle{definition}
\newtheorem{exam}[thm]{Example}
\newtheorem{fact}[thm]{Fact}

\newtheoremstyle%
 {Aside}%
 {}{}%
 {\color{purple}\itshape}
 {}%
 {\color{purple}\bfseries}%
 {\color{purple}.}%
 { }{}
\theoremstyle{Aside}

\title{The Unbounded Integrality Gap of a Semidefinite Relaxation of the Traveling Salesman Problem}
\author{Samuel C. Gutekunst and David P. Williamson }
\date{}
\begin{document}
\maketitle

\begin{abstract}
We  study a semidefinite programming relaxation of the traveling salesman problem introduced by de Klerk, Pasechnik, and Sotirov \cite{Klerk08} and show that their relaxation has an unbounded integrality gap.  In particular, we give a family of instances such that the gap increases linearly with $n$. To obtain this result, we search for feasible  solutions within a highly structured class of matrices; the problem of finding such solutions reduces to finding feasible solutions for a related linear program, which we do analytically.  The  solutions we find imply the unbounded integrality gap.  Further, they imply several corollaries that help us better understand the semidefinite program and its relationship to other TSP relaxations.  Using the same technique, we show that a more general semidefinite program introduced by de Klerk, de Oliveira Filho, and Pasechnik \cite{Klerk12} for the $k$-cycle cover problem also  has an unbounded integrality gap.    
\end{abstract}

\section{Introduction} 

The traveling salesman problem (TSP) is one of the most famous problems in combinatorial optimization.  An input to the TSP consists of a set of $n$ cities $[n]:=\{1, 2, ..., n\}$ and edge costs $c_{ij}$ for each pair of distinct $i, j \in [n]$ representing the cost of traveling from city $i$ to city $j$.  Given this information, the TSP is to find a minimum-cost tour visiting every city exactly once.  Throughout this paper, we implicitly assume that the edge costs are \emph{symmetric} (so that $c_{ij}=c_{ji}$ for all distinct $i, j\in [n]$) and \emph{metric} (so that $c_{ij}\leq c_{ik}+c_{kj}$ for all distinct $i, j, k\in [n]$).  Hence, we  interpret the $n$ cities as vertices of the complete undirected graph $K_n$ with edge costs $c_e=c_{ij}$ for edge $e=\{i, j\}$.  In this setting, the TSP is to find a minimum-cost Hamiltonian cycle on $K_n$.

The TSP is well-known to be NP-hard.  It is even NP-hard to approximate TSP solutions in polynomial time to within any constant factor  $\alpha<\f{123}{122}$ (see Karpinski,  Lampis,  and Schmied \cite{Karp15}).  For the  general TSP (without any assumptions beyond metric and symmetric edge costs), the state of the art approximation algorithm remains Christofides' 1976 algorithm \cite{Chr76}.  The output of Christofides' algorithm is at most   a factor of $\f{3}{2}$ away from the optimal solution to any TSP instance.

A broad class of approximation algorithms begin by relaxing the set of Hamiltonian cycles.   The prototypical example is the subtour elimination linear program (also referred to as the Dantzig-Fulkerson-Johnson relaxation \cite{Dan54} and the Held-Karp bound \cite{Held70}, and which we will refer to as the subtour LP). Let $V=[n]$ denote the set of vertices in $K_n,$ and let $E$ denote the set of edges in $K_n$.  For $S\subset V$, denote the set of edges with exactly one endpoint in $S$ by $\delta(S):=\{e=\{i, j\}: |\{i, j\}\cap S|=1\}$ and let $\delta(v):=\delta(\{v\}).$  The subtour elimination linear programming relaxation of the TSP is:
$$\begin{array}{l l l}
\min & \sum_{e\in E} c_e x_e & \\
\text{subject to} & \sum_{e\in \delta(v)} x_e = 2, & v=1, \ldots, n \\
& \sum_{e\in \delta(S)} x_e \geq 2, & S\subset V: S\neq \emptyset, S\neq V \\
&0\leq x_e \leq 1, & e=1, \ldots, n.
\end{array} $$
The constraints  $\sum_{e\in \delta(v)} x_e = 2$ are known as the degree constraints, while the constraints $\sum_{e\in \delta(S)} x_e \geq 2$ are known as the subtour elimination constraints.  Wolsey \cite{Wol80} and Shmoys and Williamson \cite{Shm90} show that solutions to this linear program are also within a factor of $\f{3}{2}$ of the optimal, integer solution to the TSP. 

Instead of linear programming relaxations, another approach is to  consider relaxations that are semidefinite programs (SDPs).   This avenue is considered by Cvetkovi{\'c}, {\v{C}}angalovi{\'c}, and Kova{\v{c}}evi{\'c}-Vuj{\v{c}}i{\'c} \cite{Cvet99}.  They introduce an SDP relaxation that searches for solutions that meet the degree constraints and that are at least as connected as a cycle  with respect to algebraic connectivity (see Section \ref{sec:SDPs}).  Goemans and Rendl \cite{Goe00}, however, show that the SDP relaxation of Cvetkovi{\'c} et al.\ \cite{Cvet99} is weaker than the subtour LP in the following sense: any solution to the subtour LP implies an equivalent feasible solution for the SDP of Cvetkovi{\'c} et al.\ of the same cost.  Since both optimization problems are minimization problems, the optimal  value for the SDP of Cvetkovi{\'c} et al.\ cannot be closer than the optimal solution  of the subtour LP to the optimal solution to the TSP.

More recently, de Klerk, Pasechnik, and Sotirov \cite{Klerk08} introduced another SDP relaxation of the TSP.  This SDP can be motivated and derived through a   general framework for SDP relaxations based on the theory of association schemes (see de Klerk, de Oliveira Filho, and Pasechnik \cite{Klerk12}).   Moreover, de Klerk et al.\ \cite{Klerk08} show computationally that this new SDP is incomparable to the subtour LP: there are cases for which their SDP provides a closer approximation to the TSP than the subtour LP and vice versa!  Moreover, de Klerk et al.\ \cite{Klerk08} show that their SDP is stronger than the earlier SDP of Cvetkovi{\'c} et al.\  \cite{Cvet99}: any feasible solution for the SDP of  de Klerk et al.\ \cite{Klerk08} implies a feasible solution for the SDP of Cvetkovi{\'c} et al.\  \cite{Cvet99} of the same cost.

We analyze the SDP relaxation of de Klerk et al.\ \cite{Klerk08}; our main result is that the integrality gap of this SDP is unbounded.  To show this result, we introduce a family of instances corresponding to a \emph{cut semimetric}: a subset $S\subset V$ such that $c_{ij}=1$ if $\{i, j\}\in \delta(S)$ and $c_{i j}=0$ otherwise.  We will take $|S|=\f{n}{2}.$  Equivalently, $n/2$ of the cities are located at the point $(0)\in \R^1$, the remaining $n/2$ cities are located at $(1)\in\R^1$, and the cost $c_{ij}$ is the Euclidean distance between the locations of city $i$ and city $j$. We show that for these instances the integrality gap grows linearly in $n$. The feasible solutions we introduce to bound the integrality gap, moreover, have the same algebraic connectivity as a Hamiltonian cycle $n$ vertices, even though their cost becomes arbitrarily far from that of a Hamiltonian cycle (see Section \ref{sec:SDPs}) as $n$ grows. 

%

 We introduce the SDP of de Klerk et al.\ \cite{Klerk08} in Section \ref{SDPTSP}.  In Section \ref{Idea} we motivate and prove our result.  The crux of our argument involves exploiting the symmetry of the instances we introduce.  We consider a candidate class of solutions to the SDP respecting this symmetry and show that  members of this class are feasible solutions to the SDP if and only if they are feasible solutions for a simpler linear program, whose constraints enforce certain positive semidefinite inequalities.  We then analytically find solutions to this linear program, and show that these solutions imply the unbounded integrality gap.    Next, in Section \ref{cors}, we discuss several corollaries of our main result.  These corollaries shed new light on how the SDP relates to the subtour LP as well as  to the earlier SDP of Cvetkovi{\'c} et al.\  \cite{Cvet99}. In Section \ref{sec:kcyc}, we apply our technique for showing that the integrality gap is unbounded to  a  generalization of the SDP of  de Klerk et al.\ \cite{Klerk08} for the minimum-cost $k$-cycle cover problem; when $k=1$, this problem is exactly the same as the TSP.  This more general SDP was introduced in de Klerk et al.\ \cite{Klerk12}, and we show that it also has an unbounded integrality gap.
 
 This work is related in spirit to Goemans and Rendl \cite{Goe99}, who study how to solve SDPs arising from association schemes using a linear program.  Specifically, they show that an SDP of the form $$\max \langle M_0, X\rangle \text{ s.t. } \langle M_j, X\rangle =b_j \text{ for }j=1, ..., m, \hspace{5mm} X\succeq 0,$$ where the $M_j$ are fixed, input matrices forming an association scheme, can be solved using a linear program.  Like Goemans and Rendl \cite{Goe99}, the SDP we study is related to an association scheme and we obtain a result using a linear program.   In contrast, however, to having input matrices that form an association scheme, the SDP we analyze seeks solutions that satisfy many properties of a certain, fixed association scheme (in particular, de Klerk et al.\ \cite{Klerk12} shows that the constrains of the SDP are satisfied by the association scheme corresponding to cycles; see Section \ref{SDPTSP}).  Moreover, we only use a linear program to find feasible solutions to this SDP that are sufficient to imply an unbounded integrality gap: this SDP does not in general reduce to the LP we use.

\section{A Semidefinite Programming Relaxation of the TSP}\label{SDPTSP}
\subsection{Notation and Preliminaries}
Throughout this paper we will use standard  notation from linear algebra.  We use $J_m$ and $I_m$ to  denote the all-ones and identity matrices in $\R^{m\times m},$ respectively.  When clear from context, we   suppress the dependency on the dimension and just write $J$ and $I$. We denote by $e$ the column vector of all ones, so that $J=ee^T$.  Also, we use $S^m$ for the set of real, symmetric matrices in $\R^{m\times m}$ and  $\otimes$ to denote the Kronecker product of matrices.  $A\succeq B$ denotes that $A-B$ is a positive semidefinite (PSD) matrix (we will generally have $A, B$ symmetric, in which case positive semidefiniteness is equivalent to all eigenvalues of $A-B$ being nonnegative).  The trace of a matrix $A$, denoted $\tr(A)$, is the sum of its diagonal entries so that for $A, B \in S^m$, $\tr(AB) = \sum_{i=1}^m \sum_{j=1}^m A_{ij}B_{ij}.$  $A\geq 0$ means that each entry of of matrix $A$ is nonnegative.

Our main result addresses the \emph{integrality gap} of a relaxation, which represents the worst-case ratio of the original problem's optimal solution to the relaxation's optimal solution.  We are specifically interested in the  gap of the SDP of de Klerk et al.\ \cite{Klerk08}; we will refer to this SDP as simply ``the SDP'' throughout.  Let $C$ denote a matrix of edge costs, so that $C_{ij}=C_{ji}=c_{ij}$ and $C_{ii}=0$. Let $\SDPOPT(C)$ and $\TSPOPT(C)$ respectively denote the optimal solutions to the SDP and to the TSP for a given matrix of costs $C$.  The integrality gap is then $$\sup_{C} \f{\TSPOPT(C)}{\SDPOPT(C)},$$ where we take the supremum over all valid cost matrices $C$ (those whose constituent costs are metric and symmetric).  This ratio is bounded below by 1, since the SDP is a relaxation of the TSP; we re-derive this fact in Section \ref{SDPFacts}.  We will show that the ratio cannot be bounded above by any constant.  In contrast, the results we noted previously about the subtour  LP imply that its integrality gap is bounded above by $\f{3}{2}.$

Throughout the remainder of this paper we will take $n$ to be even and let $d= \f{n}{2}.$ We use $A\backslash B$ for set minus notation, so that $A\backslash B=\{a\in A: a\notin B\}.$  We take $x\in \R^{\ch{n}{2}}$ to mean that $x$ is a vector whose entries are indexed by the edges of $K_n$.

\subsection{Facts about the SDP}\label{SDPFacts}
The SDP introduced by de Klerk et al.\ \cite{Klerk08} uses  $d$ matrix variables $X^{(1)}, ..., X^{(d)} \in \R^{n\times n}$, with the cost of a solution depending  only on $X^{(1)}.$   It is:
\begin{equation}\label{eq:SDP}
\begin{array}{l l l}
\min & \f{1}{2} \tr\left(CX^{(1)}\right) & \\
\text{subject to} & X^{(k)} \geq 0, & k=1, \ldots, d \\
& \sum_{j=1}^d X^{(j)} = J-I, & \\
& I + \sum_{j=1}^d \cos \left(\f{2\pi jk}{n}\right) X^{(j)} \succeq 0, & k=1, \ldots, d \\
& X^{(k)} \in S^n, & k=1, \ldots, d.
\end{array} \end{equation}
Both de Klerk et al.\ \cite{Klerk08} and de Klerk et al.\ \cite{Klerk12}  show that this is a relaxation of the TSP by showing that the following solution is feasible: for a simple, undirected graph $G$, let $A_k(G)$ be the \emph{$k$-th distance matrix}: the matrix with $i,j$-th entry equal to $1$ if and only if the shortest path between vertices $i$ and $j$ in $G$ is of distance $k,$ and equal to 0 otherwise.  Let $\mathcal{C}_n$ be a cycle  of length $n$ (i.e., any Hamiltonian cycle on $[n]$).  The solution  where $X^{(k)} = A_k(\mathcal{C}_n)$ for $k=1, ..., d$ is feasible for the SDP (see Proposition \ref{prop:relax}).  Hence, the optimal integer solution to the TSP has a corresponding feasible solution to the SDP. That SDP solution has the same value as the optimal integer solution to the TSP:  each edge $e=\{i, j\}$ is represented twice in $X^{(1)}$ as both $X^{(1)}_{ij}$ and $X^{(1)}_{ji},$ but this is accounted for by the factor $\f{1}{2}$ in the objective function. 

These solutions are shown to be feasible in de Klerk et al.\ \cite{Klerk08} by noting that the $A_k(\mathcal{C}_n)$ form an association scheme and are therefore simultaneously diagonalizable.  This allows for the positive semidefinite inequalities to be verified after computing the eigenvalues of each $A_k(\mathcal{C}_n)$.  A more systematic approach is taken in de Klerk et al.\ \cite{Klerk12}, where they introduce general results about association schemes.  The constraints of the SDP then represent an application of these results to a specific association scheme: that of the distance matrices $A_k(\mathcal{C}_n)$. 
%
We begin by  providing a new, direct proof that the SDP is a relaxation of the TSP.
\begin{prop}[de Klerk et al.\ \cite{Klerk08}]\label{prop:relax}
Setting  $X^{(j)} = A_j(\mathcal{C}_n)$ for $j=1, ..., d$ yields a feasible solution to the SDP (\ref{eq:SDP}).
\end{prop}
We will use two lemmas in our proof.  First, the main work in our proof involves showing that the positive semidefinite inequalities from (\ref{eq:SDP}) hold.  We do so by noticing that $ I+\sum_{j=1}^d \cos \left(\f{2\pi jk}{n}\right)  A_j(\mathcal{C}_n) $ has a very specific structure: that of a \emph{circulant} matrix. A circulant matrix is a matrix of the form
$$M=\begin{pmatrix} m_0 & m_1 & m_2 & m_3 & \cdots & m_{n-1} \\ m_{n-1} & m_0 & m_1 & m_2 & \cdots & m_{n-2} \\ m_{n-2} & m_{n-1} & m_0 & m_1 & \ddots & m_{n-3} \\ \vdots & \vdots & \vdots & \vdots & \ddots & \vdots \\ m_1 & m_2 & m_3 & m_4 & \cdots & m_0 \end{pmatrix} = \left(m_{(s-t) \text{ mod } n}\right)_{s, t=1}^n.$$
  The eigenvalues of circulant matrices are well understood, which will allow us to show that $I + \sum_{j=1}^d \cos \left(\f{2\pi jk}{n}\right)  A_j(\mathcal{C}_n)$ is a positive semidefinite matrix for each $k$ by  computing  the eigenvalues of that linear combination.  In particular:
\begin{lm}[Gray \cite{Gray06}]\label{lm:circ}
The circulant matrix $M= \left(m_{(s-t) \text{ mod } n}\right)_{s, t=1}^n$ has eigenvalues
 $$\lambda_t(M) = \begin{cases} \sum_{s=0}^{n-1} m_s e^{-\f{2\pi st  \sqrt{-1}}{n}}, & \text{ if } t=1, ..., n-1 \\  \sum_{s=0}^{n-1} m_s, & \text{ if }t=n.\end{cases}$$
\end{lm}
\noindent This is the only section where we will work with imaginary numbers, and to avoid ambiguity with index variables, we explicitly write $\sqrt{-1}$ and reserve $i$ and $j$ as index variables.

Our second lemma is a trigonometric identity that we will use repeatedly in later proofs:
\begin{lm}\label{lem:trig}
Let $n$ be even and $0<k<n$ be an integer.  Then $$\sum_{j=1}^d \cos\left(\f{2\pi jk}{n}\right) = \f{-1+(-1)^k}{2}.$$  
\end{lm}
\begin{proof}
Our identity is a consequence of Lagrange's trigonometric identity (see, e.g., Identity 14 in Section 2.4.1.6 of Jeffrey and Dai \cite{Jeff08}), which states, for $0<\th<2\pi,$ that $$\sum_{j=1}^m \cos(j \theta) = -\f{1}{2} + \f{\sin\left(\left(m+\f{1}{2}\right)\theta\right)}{2\sin\left(\f{\theta}{2}\right)}.$$  Taking $\theta = \f{2\pi k}{n}$ and using $n=2d$, we obtain:
\begin{align*}
\sum_{j=1}^d \cos\left(\f{2\pi k}{n} j\right) &= -\f{1}{2} + \f{\sin\left( \pi k+ \f{\pi k}{n}\right)}{2\sin \f{\pi k}{n}} \\ 
&=- \f{1}{2} + (-1)^k \f{1}{2},
\end{align*}
where we recall that $\sin(\pi+\theta)=-\sin(\theta).$ \hfill
\end{proof}
\noindent Notice that when $k=0$ or $k=n$, the sum is $d$.

\begin{proof}[Proof (of Proposition \ref{prop:relax})]
We first remark that each $A_j(\mathcal{C}_n)$ is a nonnegative symmetric matrix.  Moreover,  $\sum_{j=1}^d A_j(\mathcal{C}_n) = J-I.$  This follows because, in $\mathcal{C}_n$, the shortest path between any pair of distinct vertices $u, v \in [n]$ is a unique element $s$ of the set $[d]$.  
 Hence, exactly one of the terms in the sum $\sum_{j=1}^d A_j(\mathcal{C}_n)$ has a one in its $u, v$ entry, and all other terms have a zero. The diagonals of each $ A_j(\mathcal{C}_n)$ consist of all zeros, since the shortest path from vertex $i$ to itself has length $0\notin [d]$.

Now for any fixed $k\in[d]$ we compute the eigenvalues of the matrix $$M:=I+\sum_{j=1}^d \cos \left(\f{2\pi jk}{n}\right)  A_j(\mathcal{C}_n).$$ First, suppose the vertices are labeled so that the cycle $\mathcal{C}_n$ is $1, 2, 3, ..., n-1, n, 1.$  We will later note why this is without loss of generality.

 Then $M$ is circulant with, for $j=1, ..., d$, entries $m_j$ and $m_{n-j}$  given exactly by the coefficient of the $j$-th term in the sum.  Namely:
$$m_0 =1, \hspace{5mm} m_d = \cos\left(\f{2\pi kd}{n}\right), \hspace{5mm} m_j=m_{n-j}=\cos\left(\f{2\pi jk}{n}\right), j=1,..., d-1.$$  We can directly compute the $t$-th eigenvalue of $M$ using Lemma \ref{lm:circ}.  Our later proofs will include similar computations, so we pay particular emphasis to the details of our algebraic manipulation.  For $t=1, ..., n-1$, the  $t$-th eigenvalue of $M$ is:
\begin{align*}
\lambda_t(M) &= \sum_{s=0}^{n-1} m_s e^{-\f{2\pi st  \sqrt{-1}}{n}} \\
&= 1 +  \cos\left(\f{2\pi kd}{n}\right)e^{-\f{2 \pi  dt  \sqrt{-1}}{n}}+  \sum_{s=1}^{d-1 } \cos\left(\f{2\pi sk}{n}\right)\left( e^{-\f{2\pi st \sqrt{-1}}{n}}+e^{-\f{2\pi (n-s)t \sqrt{-1}}{n}}\right),
\intertext{where we have first written  the terms when $s=0$ and $s=d$.  We rewrite terms so that our sum is to $d$ and simplify exponentials:}
&= 1 -  \cos\left(\f{2\pi kd}{n}\right)e^{\f{2 \pi  dt \sqrt{-1}}{n}}+  \sum_{s=1}^{d } \cos\left(\f{2\pi sk}{n}\right)\left( e^{-\f{2\pi st \sqrt{-1}}{n}}+e^{\f{2\pi st \sqrt{-1}}{n}}\right) \\
&= 1 -  (-1)^k (-1)^t+  2\sum_{s=1}^{d } \cos\left(\f{2\pi sk}{n}\right) \cos\left(\f{2\pi st}{n}\right).
\intertext{Recalling the product-to-sum identity for cosines (that $2\cos(\theta)\cos(\phi)=\cos(\theta+\phi)+\cos(\theta-\phi)$), we get}
&= 1 -  (-1)^{k+t} +  \sum_{s=1}^{d } \cos\left(\f{2\pi s}{n}(k+t)\right) +  \sum_{s=1}^{d } \cos\left(\f{2\pi s}{n}(k-t)\right).
\intertext{Using Lemma \ref{lem:trig} and  $(-1)^{k+t}=(-1)^{k-t}$:}
&= \begin{cases}
1 - (-1)^{2d} + 2d, & \text{ if } k=t=d\\
 -\f{1}{2} + (-1)^{k+t}\f{1}{2}+ d, & \text{ if } k\neq d, t \in \{k, n-k\} \\
1 - (-1)^{k+t}  -\f{1}{2}+ (-1)^{k+t}\f{1}{2} -\f{1}{2}+ (-1)^{k-t}\f{1}{2}, & \text{ else}
\end{cases}\\
&=  \begin{cases}
2d, &\text{ if }  k=t=d\\
 d, & \text{ if } k\neq d, t \in \{k, n-k\} \\
0, & \text{ else}.
\end{cases}
\end{align*} The eigenvalue $\lambda_n$ is:
\begin{align*}
\lambda_n(M) &= \sum_{s=0}^{n-1} m_s  \\
&= 1 - \cos\left(\f{2\pi kd}{n}\right) +2 \sum_{s=1}^d \cos\left(\f{2\pi sk}{n}\right)\\
&= 1 - (-1)^k  -1 + (-1)^k\\
&= 0.
\end{align*}
The matrix $M$ thus has all nonnegative eigenvalues, so the positive semidefinite constraints hold for each $k\in\{1, ..., d\}.$  

Finally, we note that our assumption that the cycle $\mathcal{C}_n$ is $1, 2, 3, ..., n-1, n, 1$ was without loss of generality: we can replace the $ A_j(\mathcal{C}_n)$ with $P^T A_j(\mathcal{C}_n)P=P^{-1} A_j(\mathcal{C}_n)P$ for a permutation matrix $P$ that permutes the labels of the vertices so that the cycle is $1, 2, 3, ..., n-1, n, 1.$  Then $M$ and $P^{-1}MP$ are similar matrices and share the same spectrum.  Thus $M$ is positive semidefinite if and only if $P^{-1}MP$ is positive semidefinite; $P^{-1}MP$ is the circulant matrix above, with $$m_0 =1, \hspace{5mm} m_d = \cos\left(\f{2\pi kd}{n}\right), \hspace{5mm} m_j=m_{n-j}=\cos\left(\f{2\pi jk}{n}\right), j=1,..., d-1,$$ and thus both $P^{-1}MP$ and $M$ are positive semidefinite.\hfill
\end{proof} 
We briefly remark that de Klerk et al.\ \cite{Klerk08} also use the eigenvalue properties of circulant matrices in proving that the SDP is a relaxation of the TSP.  They use the fact that each individual $A_k(\mathcal{C}_n)$ is circulant to compute the eigenvalues of each $A_k(\mathcal{C}_n),$ while we use the fact that the linear combinations of those matrices denoted above by $M$ is circulant.

\section{The Unbounded Integrality Gap}\label{Idea}
To show that the SDP has an arbitrarily bad integrality gap, we demonstrate a family of instances of edge costs for which we can upper bound the SDP's objective value.  We consider an instance with two groups of $n/2$ vertices.  The costs associated to intergroup edges will be expensive (1), while the costs of intragroup edges, negligible (0).  As noted in the introduction, this instance is equivalent to both a cut semimetric and an instance where the costs are given by  Euclidean distances in $\R^1.$   Explicitly, we will use the cost matrix
$$\hat{C} := \begin{pmatrix} 
0 &\cdots  &0 & 1 & \cdots & 1 \\
\vdots & \ddots & \vdots & \vdots  & \ddots  & \vdots  \\
0 &\cdots  &0 & 1 & \cdots & 1\\
1 &\cdots  &1 & 0 & \cdots & 0 \\
\vdots & \ddots & \vdots & \vdots  & \ddots  & \vdots  \\
1 &\cdots  &1 & 0 & \cdots & 0 \\
\end{pmatrix} = \begin{pmatrix} 0 & 1 \\ 1 & 0 \end{pmatrix} \otimes J_{d}.$$
Notice that the edge costs embedded in this matrix are metric.

Throughout this paper, we reserve $U$ and $W$ to refer to the two groups of vertices, so that $|U|=|W|=d$ and $V=U\cup W$.  In a Hamiltonian cycle $\delta(U) \geq 2,$ so that any feasible solution to the TSP must use the expensive intergroup edges at least twice. 
 We can achieve a tour costing $2$ with a tour that starts in $U$, goes through all the vertices in $U$, crosses to $W$, goes through the vertices in $W$, and then returns to $U$.    
Hence $\TSPOPT(\hat{C}) = 2.$

We state our main result:
\begin{thm}\label{MainThm} $$\SDPOPT(\hat{C}) \leq \f{\pi^2}{2n} \TSPOPT(\hat{C}).$$ 
\end{thm}
As a consequence:
\begin{cor}\label{MainCor}
The SDP has an unbounded integrality gap.  That is, there exists no constant $\al >0$ such that $$\f{\TSPOPT(C)}{\SDPOPT(C)} \leq \al $$ for all cost matrices $C$.
\end{cor}
To prove this theorem, we construct a family of feasible SDP solutions whose cost becomes arbitrarily small as $n$ grows.  We will specifically search for solutions respecting the symmetry of $\hat{C}$: matrices $X^{(j)}$ that place a weight of $a_j$ on each intragroup edge and a weight of $b_j$ on each intergroup edge.  
 Moreover, we choose\footnote{Note that de Klerk et al.\ \cite{Klerk08} actually show that every feasible solution must satisfy $X^{(i)}e=2e$ for $i=1,. ..., d-1$ and $X^{(i)}e=e$ for $i=d$ (when $n$ is even).  The fact that \emph{every} feasible solution matches these row sums is not something we will need, though we implicitly use it to inform the  solutions we search for.   We provide an alternative, direct proof that all feasible solutions must satisfy these row sums in the appendix in Theorem \ref{thm:apthm}.
} the $b_i$ so as to enforce that the row sums of the $X^{(j)}$ match those of the distance matrices $A_j(\mathcal{C}_n)$ introduced earlier: $X^{(j)}e = A_j(\mathcal{C}_n)e= 2e$ for $j=1, ..., d-1$ and $X^{(d)}e=A_d(\mathcal{C}_n)e=e.$ Since every vertex is incident to $d-1$ edges in its group (with weight $a_i$) and $d$ edges in the other group (with weight $b_i$), we have 
$$(d-1)a_i + db_i = \begin{cases} 2, & \text{ if } i=1, ..., d-1 \\ 1,& \text{ if } i=d.\end{cases}$$  Rearranging for the $b_i$ lets us express the $j$-th solution matrix of this form as \begin{equation}X^{(j)}  =  \left( \begin{pmatrix} a_j & b_j \\ b_j & a_j \end{pmatrix} \otimes J_d \right) - a_j I_n, \hspace{5mm}  b_j  = \begin{cases} \f{4}{n} - \left(1-\f{2}{n}\right)a_j, & \text { if }j=1, ..., d-1 \\  \f{2}{n} - \left(1-\f{2}{n}\right)a_j, & \text{ if } j=d,\end{cases} \label{eq:solnstructure}\end{equation}  where we subtract $a_jI_n$ so that the diagonal is zero.  The cost of such a solution is entirely determined by the  $\left(n/2\right)^2$ intergroup edges, each of cost $b_1$.  Each edge is accounted for twice in $\tr(\hat{C}X^{(1)}),$ but the objective scales by $1/2$, so the cost of this solution is
$$ \left( \f{n}{2}\right)^2 b_1 .$$   
Theorem \ref{MainThm} then will follow from the claim below.
\begin{cm}  Choosing the  parameters 
$$a_i = \f{2}{n-2} \left(\cos\left(\f{\pi i}{d}\right)+1\right), \hspace{5mm} i=1, ..., d,$$ so that $$ b_i  = \begin{cases} \f{2}{n}\left(1-\cos\left(\f{\pi i}{d}\right)\right), & \text { if }i=1, ..., d-1 \\ \f{2}{n}, & \text{ if } i=d, \hspace{5mm} i=1, ..., d,\end{cases}$$
  leads to a feasible SDP solution with matrices $X^{(j)}$ as given in Equation (\ref{eq:solnstructure}). 
\end{cm}
\noindent In particular 
$b_1 = \f{2}{n}\left(1-\cos\left(\f{\pi}{d}\right)\right).$ Basic facts from calculus will show that this is roughly $\f{1}{n^3},$ so that the cost of our solution is $(n/2)^2b_1$ is roughly $\f{1}{n},$ which gets arbitrarily small with $n$.
%

The main work in proving this claim involves showing that the $X^{(j)}$ satisfy the PSD  constraints. We first characterize the choices of the $a_i$ that lead to feasible SDP solutions of the form in Equation (\ref{eq:solnstructure}); this is done in Section \ref{SDPtoLP}.  There we exploit the structure of matrices in the form of Equation (\ref{eq:solnstructure}) to write the PSD constraints on the $X^{(j)}$ as linear constraints on the $a_i$; these linear constraints will imply that all eigenvalues of  the term $I + \sum_{i=1}^d \cos \left(\f{2\pi ik}{n}\right) X^{(i)}$ are nonnegative.  To finish proving the claim, in Section \ref{FindingSolns} we show that the claimed $a_i$  are indeed feasible. 


\subsection{Finding Structured Solutions to the SDP via Linear Programing}\label{SDPtoLP}
In this section we prove the following:
\begin{prop}\label{prop:SDPLP}
For the SDP, finding a minimum-cost feasible solution   of the form $$X^{(j)}  =  \left( \begin{pmatrix} a_j & b_j \\ b_j & a_j \end{pmatrix} \otimes J_d \right) - a_j I_n \hspace{5mm} \text{ where } \hspace{5mm}  b_j  = \begin{cases} \f{4}{n} - \left(1-\f{2}{n}\right)a_j, & \text { if }j=1, ..., d-1 \\  \f{2}{n} - \left(1-\f{2}{n}\right)a_j, & \text{ if } j=d,\end{cases} $$ for $j=1, ..., d$ is equivalent  to solving the following linear program:
$$ \normalfont
\begin{array}{llll}
\max & a_1 &  \\
\text{subject to} & \sum_{i=1}^d    \cos\left(\f{2\pi ik}{n}\right)  a_i & \geq -\f{2}{n-2}, & k=1, ..., d\\
 & \sum_{i=1}^d    \cos\left(\f{2\pi ik}{n}\right)  a_i & \leq 1, &  k=1, ..., d\\
& \sum_{i=1}^d a_i &= 1 \\
& a_i & \leq \f{4}{n-2},  & i=1, ..., d-1\\
& a_d& \leq \f{2}{n-2}\\
& a_i& \geq 0,  & i=1, ..., d.
\end{array} 
$$
\end{prop}

\begin{proof} 
First we notice that maximizing $a_1$ is equivalent to minimizing $b_1,$ which is in turn equivalent to minimizing the cost $\left(\f{n}{2}\right)^2b_1$ of the SDP solution. The $X^{(i)}$ are nonnegative if and only if  $a_i \geq 0, b_i\geq 0,$ for $i=1, ..., d.$  The constraints $a_i\geq 0$ are explicit in the linear program, and $b_i\geq 0$ is equivalent to  $a_i  \leq \f{4}{n-2},   i=1, ..., d-1$ and $ a_d \leq \f{2}{n-2}.$  Finally, the constraint that the $X^{(j)}$ to sum to $J-I$ is equivalent to $ \sum_{i=1}^{d} a_i = 1$ and $\sum_{i=1}^d b_i=1.$   However,  $\sum_{i=1}^d b_i=1$ follows from requiring $\sum_{i=1}^{d} a_i = 1$:
\begin{align*}
\sum_{i=1}^d b_i &= \sum_{i=1}^{d-1} \left(\f{4}{n}-\left(1-\f{2}{n}\right)a_i\right) + \left(\f{2}{n}-\left(1-\f{2}{n}\right)a_i\right) \\
&= (d-1)\f{4}{n} + \f{2}{n} -\left(1-\f{2}{n}\right)\sum_{i=1}^d a_i\\
&=2 - \f{2}{n} - \left(1-\f{2}{n}\right) \\
&= 1.
\end{align*}

It remains to show that the $k$-th SDP constraint is equivalent to 
$$ -\f{2}{n-2}\leq \sum_{i=1}^d    \cos\left(\f{2\pi ik}{n}\right)  a_i \leq 1, \hspace{5mm}  k=1, ..., d.$$
The $k$-th SDP constraint is:
$$I + \sum_{i=1}^d \cos\left(\f{2\pi ik}{n}\right) X^{(i)} \succeq 0.$$  Using properties of the Kroenecker product (see Chapter 4 of Horn and Johnson \cite{Hor91}) and the structure of our $X^{(j)}$, we simplify this:
\begin{align*}
I_n + \sum_{i=1}^d \cos\left(\f{2\pi ik}{n}\right) X^{(i)} &= I_n +  \sum_{i=1}^d \cos\left(\f{2\pi ik}{n}\right) \left( \left( \begin{pmatrix} a_i & b_i \\ b_i & a_i \end{pmatrix} \otimes J_d \right) -a_i I_n\right)\\
&= \left(1-\sum_{i=1}^d    \cos\left(\f{2\pi ik}{n}\right)  a_i\right)I_n + \left(  \sum_{i=1}^d \cos\left(\f{2\pi ik}{n}\right) \begin{pmatrix} a_i & b_i \\ b_i & a_i \end{pmatrix} \right) \otimes J_d  \\
&= (1- a^{(k)})I_n +  \begin{pmatrix}  a^{(k)} &  b^{(k)} \\ b^{(k)} &  a^{(k)} \end{pmatrix} \otimes J_d,
\end{align*}
where $$a^{(k)} =\sum_{i=1}^d    \cos\left(\f{2\pi ik}{n}\right)  a_i, \hspace{5mm} b^{(k)} =\sum_{i=1}^d    \cos\left(\f{2\pi ik}{n}\right)  b_i$$ depend on the full sequences  $a_1, ..., a_d, b_1, ..., b_d$ and on $k$.

To explicitly write the eigenvalues of the $k$-th SDP constraint, we use several helpful facts from linear algebra.
\begin{fact} \hfill
\begin{itemize}
 \item The $pq$ eigenvalues of $A\otimes B$ with $A\in \R^{p\times p}$ and $B\in \R^{q \times q}$ are $\lambda_i(A) \lambda_j(B)$ for $i=1, ..., p$ and $j=1, ..., q$.  See Theorem 4.2.12 in Chapter 4 of Horn and Johnson  \cite{Hor91}.
 \item The rank one matrix $J_d = ee^T$, with $e$ of dimension $d$, has one eigenvalue $d$ corresponding to eigenvector $e$, and all other eigenvalues are zero.  (Choose, e.g., any $d-1$ linearly independent vectors that are orthogonal to $e$.)
 \item  $\lambda(A)$ is an eigenvalue of $A$ with eigenvector $v$ if and only if $\lambda(A)+c$ is an eigenvalue of $A+cI$ with eigenvector $v$. This follows by direct computation.
\item The eigenvalues of $\begin{pmatrix} a & b \\ b & a \end{pmatrix}$ are $a+b$ and $a-b$ with respective eigenvectors  $\begin{pmatrix} 1 \\ 1 \end{pmatrix}$ and   $\begin{pmatrix} 1 \\ -1 \end{pmatrix}.$   
 \end{itemize}
\end{fact}
\noindent From these facts, we obtain that the eigenvalues of $I + \sum_{i=1}^d \cos\left(\f{2\pi ik}{n}\right) X^{(i)}$ are: $$1 - a^{(k)}, \hspace{5mm}  1-a^{(k)} + \f{n}{2} \left(a^{(k)}+b^{(k)}\right), \hspace{5mm} \text{ and } 1-a^{(k)} + \f{n}{2} \left(a^{(k)}-b^{(k)}\right).$$   For example, $1-a^{(k)}$ has multiplicity $n-2$.  It corresponds to the $d-1$ zero eigenvalues of $J_d$, each of which gives rise to 2 zero eigenvalues of $\begin{pmatrix}  a^{(k)} &  b^{(k)} \\ b^{(k)} &  a^{(k)} \end{pmatrix} \otimes J_d.$

Therefore, for the $k$-th PSD  constraint to hold,  it suffices that the following three linear inequalities  hold: \begin{equation}\label{eq:lineqs} 1 - a^{(k)} \geq 0, \hspace{5mm} 1-a^{(k)} + \f{n}{2} \left(a^{(k)}+b^{(k)}\right) \geq 0,  \hspace{5mm} 1-a^{(k)} + \f{n}{2} \left(a^{(k)}-b^{(k)}\right) \geq 0. \end{equation}

We thus far have derived a system of inequalities on the $a_i, b_i$  that, if satisfied, imply a set of feasible solutions to the SDP.  We can further simplify these by writing the $b_i$ in terms of the $a_i$.  As in Proposition \ref{prop:relax}, we begin by writing the sum so that we can use Lemma \ref{lem:trig}.  We compute
%
\begin{align*}
b^{(k)} &=\sum_{i=1}^d    \cos\left(\f{2\pi ik}{n}\right)  b_i  \\
&=\left(\sum_{i=1}^{d-1}    \cos\left(\f{2\pi ik}{n}\right)  \left(\f{4}{n}-\left(1-\f{2}{n}\right)a_i\right)\right)+\cos\left(\f{2\pi dk}{n}\right)  \left(\f{2}{n}-\left(1-\f{2}{n}\right)a_d\right) \\
&=\f{4}{n} \left(\sum_{i=1}^{d}    \cos\left(\f{2\pi ik}{n}\right) \right)-\left(1-\f{2}{n}\right) \left(\sum_{i=1}^{d}    \cos\left(\f{2\pi ik}{n}\right) a_i\right)-\cos(\pi k)   \left(\f{2}{n}\right)
\intertext{Using Lemma \ref{lem:trig}:}
&=\f{4}{n}  \left(\f{-1+(-1)^k}{2}\right)-\left(1-\f{2}{n}\right)a^{(k)}-\left(-1\right)^k  \left(\f{2}{n}\right)\\
&=-\left(1-\f{2}{n}\right)a^{(k)} - \f{2}{n}.
\end{align*}
We use this relationship to simplify the second and third inequalities in Equation (\ref{eq:lineqs}) by writing them only in terms of $a^{(k)}.$  We obtain
$$1-a^{(k)} + \f{n}{2} (a^{(k)}+b^{(k)}) =1-a^{(k)}+\f{n}{2} \left(a^{(k)}-\left(1-\f{2}{n}\right)a^{(k)}-\f{2}{n}\right)=0,$$
and
$$ 1-a^{(k)} + \f{n}{2} (a^{(k)}-b^{(k)}) =1-a^{(k)} + \f{n}{2} \left(a^{(k)}+\left(1-\f{2}{n}\right)a^{(k)} + \f{2}{n}\right)= 2+(n-2)a^{(k)}.$$ 
Hence, the three inequalities in  Equation (\ref{eq:lineqs}) become
 $$-\f{2}{n-2}\leq a^{(k)} \leq 1,$$ and these inequalities are equivalent to ensuring that the $k$-th PSD constraint of the SDP in (\ref{eq:SDP}) hold.\hfill
\end{proof}
\begin{cor}\label{rangecor} 
Consider a possible solution to the SDP of the form $$X^{(j)}  =  \left( \begin{pmatrix} a_j & b_j \\ b_j & a_j \end{pmatrix} \otimes J_d \right) - a_j I_n \hspace{5mm} \text{ where } \hspace{5mm}  b_j  = \begin{cases} \f{4}{n} - \left(1-\f{2}{n}\right)a_j, & \text { if }j=1, ..., d-1 \\  \f{2}{n} - \left(1-\f{2}{n}\right)a_j, & \text{ if } j=d,\end{cases}.$$ The $k$th PSD constraint  $I + \sum_{j=1}^d \cos \left(\f{2\pi jk}{n}\right) X^{(j)} \succeq 0$ is equivalent to $-\f{2}{n-2}\leq a^{(k)} \leq 1.$
\end{cor}

\subsection{Analytically Finding Solutions to the Linear Program} \label{FindingSolns}
We now show that the following choice of the $a_i$ lead to $X^{(j)}$ that are feasible for the SDP  (\ref{eq:SDP}):
$$a_i = \f{2}{n-2} \left(\cos\left(\f{\pi i}{d}\right)+1\right), \hspace{5mm} i=1, ..., d.$$
As argued above, to show feasibility we need only verify that the constraints of the linear program in Proposition \ref{prop:SDPLP} hold.  Notice that   $-1\leq \cos\left(\pi i/d\right) \leq 1$ so that, for $i=1, ..., d-1,$ we have $0\leq a_i \leq \f{4}{n-2}.$  Moreover, $a_d=0$.  Hence we need only show that $\sum_{i=1}^d a_i = 1$ and that the $a^{(k)}$ live in the appropriate range.

\begin{cm}\label{cm:asum}
For $a_i = \f{2}{n-2} \left(\cos\left(\f{\pi i}{d}\right)+1\right),$ $$\sum_{i=1}^d a_i = 1.$$
\end{cm}

\begin{proof}
We directly compute $\sum_{i=1}^d a_i$ using Lemma \ref{lem:trig} with $k=1$.  Then:
\begin{align*}
\sum_{i=1}^d a_i &= \f{2}{n-2} \sum_{i=1}^d  \left(\cos\left(\f{\pi i}{d}\right)+1\right)\\
&= \f{2}{n-2}\left(-1 + d\right) \\
&= 1.
\end{align*} \hfill
\end{proof}

\begin{cm}\label{cm:aj}
For $a_i = \f{2}{n-2} \left(\cos\left(\f{\pi i}{d}\right)+1\right),$ $$a^{(k)} = \begin{cases} \f{d-2}{n-2}, & \text{ if } k=1 \\ -\f{2}{n-2}, & \text{ otherwise}. \end{cases}$$ 
\end{cm}

\begin{proof}
As in Proposition \ref{prop:relax}, we use the product-to-sum identity for cosines and then do casework using Lemma \ref{lem:trig}. We have:
\begin{align*}
a^{(k)} &=\sum_{i=1}^d  \cos\left(\f{2\pi i k}{n}\right) a_i\\
 &=\f{2}{n-2}  \sum_{i=1}^d \left( \cos\left(\f{2\pi ik}{n}\right) +\cos\left(\f{2\pi ik}{n}\right) \cos\left(\f{\pi i}{d}\right) \right)\\
 &=\f{2}{n-2}  \sum_{i=1}^d \left( \cos\left(\f{2\pi ik}{n}\right) + \f{1}{2}\cos \left(\f{2\pi i(k+1)}{n} \right)  +  \f{1}{2}\cos \left( \f{2\pi i(k-1)}{n}\right) \right) 
\intertext{We  cannot apply Lagrange's trigonometric identity only when $k=1$, so that}
&= \begin{cases} \f{2}{n-2}  \left( \f{-1+(-1)^k}{2} + \f{-1 + (-1)^{k+1}}{4}  + \f{-1 + (-1)^{k-1}}{4}\right), & \text{ if } k>1 \\
\f{2}{n-2}  \left( -1+ 0+ \f{1}{2} d\right), & \text{ if } k=1 \end{cases} \\
&= \begin{cases} -\f{2}{n-2}, & \text{ if } k>1 \\
\f{d-2}{n-2} , & \text{ if } k=1. \end{cases} 
\end{align*} \hfill
\end{proof}

Claim \ref{cm:aj} and Corollary \ref{rangecor} now show that the claimed $a_i$ imply feasible solutions satisfying the PSD constraints.  Taken with Claim \ref{cm:asum} and Proposition \ref{prop:SDPLP}, we have that $$a_i = \f{2}{n-2} \left(\cos\left(\f{\pi i}{d}\right)+1\right), i=1, ..., d$$ is feasible for the linear program in Proposition \ref{prop:SDPLP} and therefore implies feasible solutions for the SDP  (\ref{eq:SDP}) of the form $$X^{(j)}  =  \left( \begin{pmatrix} a_j & b_j \\ b_j & a_j \end{pmatrix} \otimes J_d \right) - a_j I_n \hspace{5mm} \text{ where } \hspace{5mm}  b_j  = \begin{cases} \f{4}{n} - \left(1-\f{2}{n}\right)a_j, & \text { if }j=1, ..., d-1 \\  \f{2}{n} - \left(1-\f{2}{n}\right)a_j, & \text{ if } j=d.\end{cases} $$

\subsection{The Unbounded Integrality Gap} \label{Cons}
We are now able to prove our main theorem:

\begingroup
\def\thetheorem{\noindent {\bf Theorem \ref{MainThm}}}
\begin{thetheorem} \emph{$$\SDPOPT(\hat{C}) \leq \f{\pi^2}{2n} \TSPOPT(\hat{C}).$$} 
\end{thetheorem}
\endgroup

\begin{proof}
Earlier we saw that a feasible solution of the form in Equation  (\ref{eq:solnstructure}) had cost $\f{n^2}{4}b_1$ and $\TSPOPT(\hat{C})=2.$  Hence, assuming a feasible solution, we can bound
$$\f{\SDPOPT(\hat{C})}{\TSPOPT(\hat{C})} \leq \f{n^2 b_1}{8}.$$ We have since found a feasible solution with parameter $$a_1=  \f{2}{n-2} \left(\cos\left(\f{\pi}{d}\right)+1\right)$$ so that
$$b_1  = \f{4}{n} - \left(1-\f{2}{n}\right)\f{2}{n-2} \left(\cos\left(\f{\pi}{d}\right)+1\right) = \f{2}{n}\left(1-\cos\left(\f{\pi}{d}\right)\right).$$
Using Taylor series with remainder, 
 $$\cos\left(\f{\pi}{d}\right) = 1-\f{\pi^2}{2d^2} + \f{1}{4!}\f{\pi^4}{d^4} \cos\left(\xi_{1/d}\right) \geq 1-\f{\pi^2}{2d^2}, 
$$ 
where $\xi_{1/d} \in [0, \f{1}{d}]$

Hence, we bound:
\begin{align*}
\f{\SDPOPT(\hat{C})}{\TSPOPT(\hat{C})} &\leq \f{n^2 b_1}{8}\\
&\leq \f{n^2}{8} \f{2}{n} \left(\f{\pi^2}{2d^2} \right) \\
&= \f{\pi^2}{2n}. 
\end{align*}
 \hfill
\end{proof}

We note that, at best, the SDP is an $\mathcal{O}\left(n \right)$-approximation algorithm. Also we notice the following: 

\begin{rem}
Several hierarchies exist that strengthen convex relaxations of combinatorial optimization problems, including those of Sherali and Adams \cite{She90}, Lov{\'a}sz  and Schrijver \cite{Lov91}, and Lasserre \cite{Las01}. These hierarchies iteratively add constraints to the relaxation; after sufficiently many iterations, the surviving feasible solutions correspond exactly to convex combinations of integer solutions.   See Chlamtac and Tulsiani \cite{Chl12} for a  detailed survey.  

Cheung \cite{Che05}, for example, applies  hierarchies to show that certain feasible solutions for the subtour LP survive applying the  Lov{\'a}sz  and Schrijver hierarchy any constant number of times. In particular, those solutions violated certain constraints (2-matching inequalities) satisfied by Hamiltonian cycles.  One might analogously wonder how long our solution survives iteratively adding constraints to an appropriate linear program. $X^{(1)}$ is not feasible for the subtour LP for sufficiently large $n$, so that it trivially doesn't survive any rounds of these hierarchies applied to the subtour LP.  In contrast, it can be shown that the feasible $X^{(1)}$ we found is in the convex hull of cycle covers.  Hence our solution would survive arbitrarily many rounds of any of these hierarchies applied to linear program obtained by using only the degree constraints of the subtour LP.
\end{rem}


\section{Corollaries of Theorem \ref{MainThm}}\label{cors}
Theorem \ref{MainThm} and its proof imply several corollaries that help us better understand the SDP and its relationship to other relaxations of the TSP.  We list several corollaries in this section, first relating the SDP to the subtour LP (Sections \ref{sec:nonmon} through \ref{sec:LPSDP}), and then relating the SDP to another SDP for the TSP in Section \ref{sec:SDPs}.

%

\subsection{Non-Monotonicity of Solution Costs}\label{sec:nonmon}
We begin with the counterintuitive result that adding vertices (in a way that retains costs being metric) can arbitrarily decease the cost of some solutions to the SDP. We state this as a \emph{non-monotonicity} property that contrasts with both TSP and subtour LP solutions.  

Consider an optimization problem whose variables correspond to edges of the complete graph $K_n$ and whose input consists of edge costs and a size $n$.  Let $S\subset [n]$ be a subset of the vertices.  Let $\text{OPT}$ denote the cost of the optimal solution to the optimization problem on the full set of vertices, and let $\text{OPT}[S]$ denote the  the cost of the optimal solution to the optimization problem \emph{induced} on the set $S.$  Formally, if $C$ denotes the matrix of edge costs corresponding to the original input, then the induced problem on $S$ uses the edge cost matrix $C[S]$ defined to be the principle submatrix of $C$ obtained by deleting the rows and columns in $[n]\backslash S$.  If $\text{OPT}[S]\leq \text{OPT}$ for all possible input costs, values of $n$, and subsets $S$, we say that the the optimization property has a \emph{monotonicity} property.  

The TSP (as usual, assuming metric and symmetric edge costs) is well-known to be monotonic (this can be seen as an application of \emph{shortcutting.}  See  Section 2.4 of Williamson and Shmoys \cite{DDBook} for details of shortcutting.)  Moreover, Shmoys and Williamson \cite{Shm90} show that the subtour LP is also monotonic.  Our example shows that the SDP of de Klerk et al.\ \cite{Klerk08}, however, is not: the cost of our SDP solutions get arbitrarily small as $n$ grows, and our instance on $n'$ vertices can be viewed as an induced from a larger instance on $n>n'$ vertices.

\begin{cor}
The SDP in (\ref{eq:SDP}) is not monotonic.
\end{cor}

\subsection{The Relationship of our SDP Solutions to the Minimum Spanning Tree Polytope}\label{sec:nonMST}
The \emph{minimum spanning tree} (MST) polytope is 
$$\{z\in \R^{\ch{n}{2}} : \sum_{e\in E} x_e = n-1, \sum_{e\in E(S)} z_e \leq |S|-1 \text{ for all } S\subset V, z\geq 0\}.$$
One nice, well-known property of the subtour  LP is that any feasible solution to it, when appropriately scaled, is also feasible for the MST polytope (see, e.g., Gao \cite{Gao15} for a very similar argument).  
%
Conversely, solutions to the SDP  cannot in general be scaled to be in the MST polytope.  We show this directly using our feasible solutions\footnote{We briefly note that, if we could appropriately scale the SDP solutions to be in the MST, we would be able to bound the integrality gap by a factor of 2 by using the standard tree-doubling approximation algorithm (see, e.g., Section 2.4 of  Williamson and Shmoys \cite{DDBook}); from this observation, and the fact that we have shown that the integrality gap is unbounded, it follows that our solutions cannot be scaled to lie in the MST polytope. Here we instead chose to provide a direct proof that reveals how far our solutions are outside of the MST polytope.}.    

%


\begin{cor}
Let $x\in \R^{\ch{n}{2}}$ be defined by $x_e = X^{(1)}_{i j} = X^{(1)}_{j i}$ and denote by $E$ the set of all edges in the complete graph on $n$ vertices.  There is no suitable scaling factor $c$ such that $cx$ is in the minimum spanning tree polytope (where $c$ is allowed to be a function of $n$).
\end{cor}

\begin{proof}
Notice that $X^{(1)}e=2e$ implies that $\sum_{e\in E} x_e = n$ so that we must set $c=\f{n-1}{n}$.  Again let $U$ correspond to the set of vertices in one group.  Then there are $\ch{d}{2}$ edges in $E(U)$, each of which has is assigned weight of $a_1$ in our solution.  Hence:
\begin{align*}
\f{n-1}{n} \sum_{e\in E(U)} x_e &= \f{n-1}{n} \sum_{e\in E(U)} a_1 \\
&= \f{n-1}{n} \ch{n/2}{2} \f{2}{n-2} \left(\cos\left(\f{\pi}{d}\right)+1\right) \\
&= \f{n-1}{4} \left(\cos\left(\f{\pi}{d}\right)+1\right) \\
& \geq \f{n-1}{4} \left(2 - \f{\pi^2}{2d^2} \right) \\ 
&= d - \f{1}{2} + O\left(\f{1}{n}\right) \\
& >  |U|-1,
\end{align*}
for all $n$ sufficiently large.   \hfill
\end{proof}


\subsection{The SDP and Subtour Elimination Linear Program When $n$ is Small}\label{sec:LPSDP}
When $n=6$, our solution is
$$X^{(1)}= \begin{pmatrix} 3/4 & 1/6 \\ 1/6 & 3/4 \end{pmatrix}\otimes J_3 - \f{3}{4}I_6 .$$  Letting $U=\{1, 2, 3\}$ denote one of the two groups of vertices, we see that $\delta(U)$ has 9 edges in it, each of which is assigned a weight of $1/6$, so that the total weight crossing
$\delta(U)$ in this solution is $9*\f{1}{6} = \f{3}{2} < 2.$  This violates the subtour elimination constraint for $U$. Hence, we see that the subtour LP and SDP have distinct feasible regions when $n=6$.   
We can show, in contrast, that they are the same for $n\leq 5.$  Doing so involves  computations that are of a different spirit than what we have done so far; we defer this proof to the Appendix.   We emphasize this result because, when $n\leq 5$, it is known that the feasible region to the subtour LP consists exactly of convex combinations of Hamiltonian cycles.  See, for example, Gr{\"o}tschel and Padberg \cite{Gro86}.   Hence this result lets us  characterize the feasible region to the SDP when $n\leq 5$ as corresponding exactly to convex combinations of Hamiltonian cycles.  We state and formalize these results in  Lemma \ref{lem:simpsubtour}.

\subsection{The Relationship of our Solution to an Earlier TSP SDP}\label{sec:SDPs}
Previously we mentioned an earlier SDP relaxation for the TSP from  Cvetkovi{\'c} et al.\ \cite{Cvet99} which was shown to be weaker than the subtour LP in Goemans and Rendl  \cite{Goe00}.   This relaxation has a single matrix variable $X$ and takes the form:
\begin{equation}\label{eq:CSDP}
\begin{array}{l l l}
\min & \f{1}{2} \tr\left(CX\right) & \\
\text{subject to} & Xe = 2e \\
& X_{ii}=0,& i=1, ..., n \\
&  X_{ij} \leq 1, & i, j = 1, ..., n \\
&2I - X + \left(2-2\cos\left(\f{2\pi}{n}\right)\right) (J-I) \succeq 0\\
& X \in S^n. &
\end{array} \end{equation}
The variable $X$ can be interpreted as a weighted adjacency matrix, and the constraint that $Xe=2e$ enforces $e$ is an eigenvector of $X$ with corresponding eigenvalue $2$.  The term $2I-X$ in the PSD constraint can be interpreted as the Laplacian of  $X$:   let $G$ be a weighted, undirected graph on $n$ vertices with weighted adjacency matrix $A$.  Let $D$ be the degree matrix of $G$ (i.e., $D$ is diagonal with $D_{ii} = \sum_{j=1}^n A_{ij}$).  The Laplacian of $G$ is defined as $$L:=D-A.$$  With the interpretation of $X$ as a weighted adjacency matrix, the constraint $Xe=2e$  implies that the Laplacian corresponding to $X$ is $$L(X):=2I-X,$$ where we make the dependence on $X$ explicit.  This observation, and machinery from spectral graph theory, motivates the positive semidefinite constraint in the SDP of  Cvetkovi{\'c} et al.\ \cite{Cvet99}.  (See Spielman \cite{Spi07} for a nice introduction to spectral graph theory.)  

In more detail, let $h_n:=2-2\cos\left(\f{2\pi}{n}\right)$ so that  the positive semidefinite constraint  is $$L(X) + h_n(J-I) \succeq 0.$$   The value of $h_n$ is known to be the second smallest eigenvalue of the Laplacian of a cycle on $n$ vertices\footnote{The Laplacian of a cycle graph is also a circulant matrix, with $m_0=2$, $m_1=m_{n-1}=-1$ and $m_i=0$ otherwise.  Its eigenvalues can be directly computed using Lemma \ref{lm:circ}. }.  The second smallest eigenvalue of the Laplacian is known as the \emph{algebraic connectivity} of a graph.

The Laplacian of a weighted graph is known to be positive semidefinite (see Spielman \cite{Spi07}, which represents the Laplacian as a quadratic form), so we can write the eigenvalues of $L(X)$ as $0\leq \lambda_1\leq \lambda_2 \leq \cdots \leq \lambda_n.$   Since $X$ is symmetric, we further assume that these eigenvalues correspond to an orthogonal basis of eigenvectors $v_1, ..., v_n$ where  $v_i$ corresponds to eigenvalue $\lambda_i$.  Moreover, we can choose to let $v_1=e$ and $\lambda_1=0,$ since $Xe=2e$.  The eigenvalues of $L(X)+h_n(J-I)$ are then:
$$\lambda_1 +(n-1)h_n = (n-1)h_n, \lambda_2-h_n, ...,  \lambda_n -h_n.$$  
These follow by right-multiplying  $L(X)+h_n(J-I)$ by $v_i$ and noting that $Jv_i = ee^Tv_i = 0$ if $i\neq 1$, and $Jv_1 = ee^Te = ne.$ Since $h_n \geq 0,$  the positive semidefinite constraint enforces that $$\lambda_i - h_n \geq 0, \hspace{5mm} i=1, ..., n-1,$$ or equivalently that \begin{equation}\label{eq:eig} \lambda_2 \geq h_n.\end{equation}
Hence, the PSD constraint introduced by Cvetkovi{\'c} et al.\ \cite{Cvet99}  enforces that the algebraic connectivity of $X$ is at least $h_n$, the algebraic connectivity of a cycle on $n$ vertices. 

One might wonder if our solution $X^{(1)}$ is also feasible for the SDP of  Cvetkovi{\'c} et al.\ \cite{Cvet99}.  The answer is yes because, as mentioned earlier, as de Klerk et al.\ \cite{Klerk08} showed that any solution of (\ref{eq:SDP}) is feasible for the SDP of Cvetkovi{\'c} et al.\ \cite{Cvet99} in (\ref{eq:CSDP}).  
Hence, Theorem \ref{MainThm} implies that the SDP (\ref{eq:CSDP}) also has an unbounded integrality gap.

 Here we show the result directly for our feasible solutions, as it turns out that our solution corresponds to an instance where Equation (\ref{eq:eig}) is tight.  Thus our $X^{(1)}$ instance and cost matrix $\hat{C}$ provide an explicit example of a weighted graph that has exactly the same algebraic connectivity as a cycle and in which every vertex has degree two, but has cost arbitrarily far from a minimum-cost Hamiltonian cycle.
\begin{prop}
Taking $$X=X^{(1)} =\left( \begin{pmatrix} a_1 & b_1 \\ b_1 & a_1 \end{pmatrix} \otimes J_d \right)-a_1 I_n,$$ with $a_1=  \f{2}{n-2} \left(\cos\left(\f{\pi}{d}\right)+1\right)$ and $b_1  = \f{2}{n}\left(1-\cos\left(\f{\pi}{d}\right)\right)$ yields a feasible solution for the SDP (\ref{eq:CSDP}).  Moreover, the algebraic connectivity of $X$ is exactly that of an $n$-cycle.
\end{prop}

\begin{proof}
By construction, $X^{(1)}$ satisfies all conditions of (\ref{eq:CSDP}) except possibly that $$2I_n-X^{(1)}+h_n(J_n-I_n) \succeq 0.$$  By the argument above, it suffices to compute the second smallest eigenvalue of $2I_n-X^{(1)}$ and show that it is at least $h_n$.  The eigenvalues of $$2I_n-X^{(1)} = (2+a_1)I_n - \left( \begin{pmatrix} a_1 & b_1 \\ b_1 & a_1 \end{pmatrix} \otimes J_d \right)$$ are $2+a_1$, with multiplicity $n-2$, and $2+a_1 - d(a_1\pm b_1)$, each with multiplicity 1.  Simplifying these later eigenvalues, we have the two eigenvalues $$2+a_1-d(a_1+b_1)=0, \hspace{5mm} 2+a_1-d(a_1-b_1)=h_n.$$   Hence, the second smallest eigenvalue of $L(X^{(1)})$ is indeed $h_n$. \hfill
\end{proof}  

\begin{cor}
The SDP (\ref{eq:CSDP}) has an unbounded integrality gap.
\end{cor}

\begin{cor}
The algebraic connectivity of $X^{(1)}$ is equal to the algebraic connectivity of cycle.
\end{cor}


%
%
%

\section{The $k$-Cycle Cover Problem}\label{sec:kcyc}
In the TSP, we try to find a minimum-cost cycle that covers all verticies.  This problem is generalized in the $k$-cycle cover problem which involves finding $k$ equally sized cycles that cover all of the vertices (and assumes $n$ is divisible by $k$).  Just as in the TSP, the goal is to do so with minimum-cost.    As for the TSP, there are algorithms for finding approximate solutions with a bounded integrality gap.  Goemans and Williamson \cite{Goe95} give a 4-approximation algorithm for this problem.    

De Klerk et al.\ \cite{Klerk12} notice that the SDP (\ref{eq:SDP}) can be  modified to become a relaxation of the $k$-cycle problem by changing only the objective function.  They argue the following:
\begin{prop}  The following SDP is a relaxation of the minimum-cost $k$-cycle cover problem. \normalfont
\begin{equation}\label{eq:kSDP}
\begin{array}{l l l}
\min & \f{1}{2} \text{trace}\left(CX^{(k)}\right) & \\
\text{subject to} & X^{(j)} \geq 0, & j=1, \ldots, d \\
& \sum_{j=1}^d X^{(j)} = J-I & \\
& I + \sum_{j=1}^d \cos \left(\f{2\pi i j}{n}\right) X^{(j)} \succeq 0, & i=1, \ldots, d \\
& X^{(i)} \in S^n, & i=1, \ldots, d.
\end{array} \end{equation}
\end{prop}

\begin{proof}[Proof (from de Klerk et al.\ \cite{Klerk12})]
The proof uses exactly the same feasible solutions as Proposition \ref{prop:relax}.  The key observation is that the $k$-th distance matrix $A_k(\mathcal{C}_n)$ represents a partition of the vertices into $k$ equally sized cycles.  In particular, if $\mathcal{C}_n$ is the cycle $v_1, v_2, ..., v_n, v_1$, then $A_k(\mathcal{C}_n)$ consists of the cycles $v_i, v_{i+k}, v_{i+2k}, ..., v_{i+ (n-k)}, v_i$ for $i=1, 2, ..., k$. See, for example, Figure \ref{fig:example}.  Any $k$-cycle cover of the vertices can similarly be represented as the $k$-th distance matrix of some Hamiltonian cycle. \hfill
\end{proof}

\begin{figure}[h!t]%
    \centering
\begin{tikzpicture}[scale=0.85]
\tikzset{vertex/.style = {shape=circle,draw,minimum size=2.3em}}
\tikzset{edge/.style = {->,> = latex'}}
\tikzstyle{decision} = [diamond, draw, text badly centered, inner sep=3pt]
\tikzstyle{sq} = [regular polygon,regular polygon sides=4, draw, text badly centered, inner sep=3pt]
\node[vertex] (a) at  (3, 0) {$v_1$};
\node[vertex] (b) at  (2.6, -1.5) {$v_2$};
\node[vertex] (c) at  (1.5, -2.6) {$v_3$};
\node[vertex] (d) at  (0, -3) {$v_4$};
\node[vertex] (e) at  (-1.5, -2.6) {$v_5$};
\node[vertex] (f) at  (-2.6, -1.5) {$v_6$};
\node[vertex] (g) at  (-3, 0) {$v_7$};
\node[vertex] (h) at  (-2.6, 1.5) {$v_8$};
\node[vertex] (i) at  (-1.5, 2.6) {$v_9$};
\node[vertex] (j) at  (0, 3) {$v_{10}$};
\node[vertex] (k) at  (1.5, 2.6) {$v_{11}$};
\node[vertex] (l) at  (2.6, 1.5) {$v_{12}$};

\draw (a) -- (b);
\draw (c) -- (b);
\draw (c) -- (d);
\draw (e) -- (d);
\draw (e) -- (f);
\draw (g) -- (f);
\draw (g) -- (h);
\draw (i) -- (h);
\draw (i) -- (j);
\draw (k) -- (j);
\draw (k) -- (l);
\draw (a) -- (l);
\end{tikzpicture}
    \qquad
\begin{tikzpicture}[scale=0.85]
\tikzset{vertex/.style = {shape=circle,draw,minimum size=2.3em}}
\tikzset{edge/.style = {->,> = latex'}}
\tikzstyle{decision} = [diamond, draw, text badly centered, inner sep=3pt]
\tikzstyle{sq} = [regular polygon,regular polygon sides=4, draw, text badly centered, inner sep=3pt]
\node[vertex] (a) at  (3, 0) {$v_1$};
\node[vertex] (b) at  (2.6, -1.5) {$v_2$};
\node[vertex] (c) at  (1.5, -2.6) {$v_3$};
\node[vertex] (d) at  (0, -3) {$v_4$};
\node[vertex] (e) at  (-1.5, -2.6) {$v_5$};
\node[vertex] (f) at  (-2.6, -1.5) {$v_6$};
\node[vertex] (g) at  (-3, 0) {$v_7$};
\node[vertex] (h) at  (-2.6, 1.5) {$v_8$};
\node[vertex] (i) at  (-1.5, 2.6) {$v_9$};
\node[vertex] (j) at  (0, 3) {$v_{10}$};
\node[vertex] (k) at  (1.5, 2.6) {$v_{11}$};
\node[vertex] (l) at  (2.6, 1.5) {$v_{12}$};

\draw (a) -- (d);
\draw (d) -- (g);
\draw (g) -- (j);
\draw (j) -- (a);
\draw (b) -- (e) [dashed] ;
\draw (e) -- (h) [dashed] ;
\draw (h) -- (k) [dashed] ;
\draw (k) -- (b) [dashed] ;
\draw (c) -- (f) [thick, densely dotted];
\draw (f) -- (i) [thick, densely dotted];
\draw (i) -- (l) [thick, densely dotted];
\draw (c) -- (l) [thick, densely dotted];
\end{tikzpicture}
    \caption{The graphs corresponding to $A_1(\mathcal{C}_n)$ and $A_3(\mathcal{C}_n)$ when $n=12$.  Notice that the right graph is a 3-cycle cover, and each cycle is drawn with a different edge style.}%
    \label{fig:example}%
\end{figure}
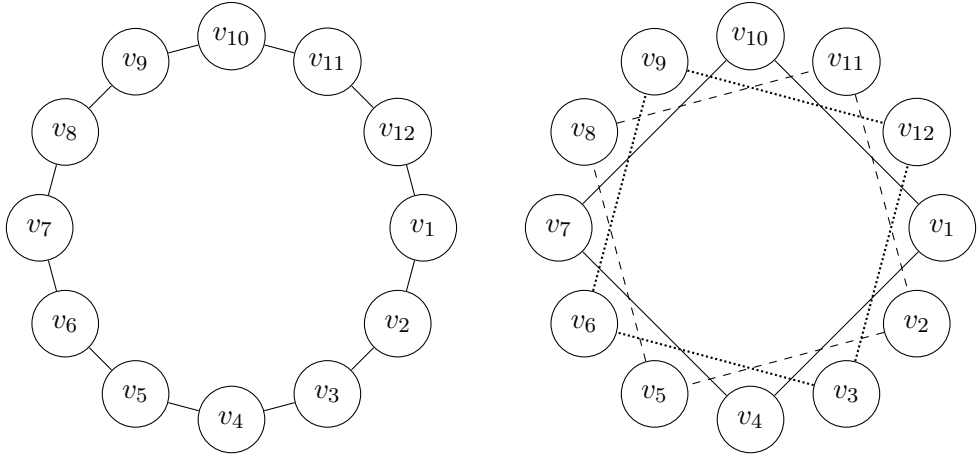

Since the SDP for the TSP is a special case of this SDP obtained by setting $k=1$, it is natural to wonder if our technique also shows that this more general SDP has an unbounded integrality gap.  Again the answer is in the affirmative. Let $\SDPOPT(C)$ and $\CYCOPT(C)$ respectively denote the optimal solutions to the SDP (\ref{eq:kSDP}) and to the $k$-cycle cover problem for a given matrix of costs $C$ and fixed $k\geq 2.$   Our earlier result generalizes as follows:
\begin{thm}\label{MainThm2}
$$\SDPOPT(\hat{C}) \leq \f{\pi^2}{n}\f{k}{k+1} \CYCOPT(\hat{C}).$$ 
\end{thm}
\begin{cor}
The SDP for the $k$-cycle cover problem has an unbounded integrality gap.  That is, there exists no constant $\al >0$ such that $$\f{\CYCOPT(C)}{\SDPOPT(C)} \leq \al $$ for all cost matrices $C$.
\end{cor}

To prove this, we modify our earlier example and consider cost matrices reflecting $k+1$ equally sized groups of vertices.  Hence, we let $n=ck(k+1)$ and will scale $n$ by scaling $c\in\N$ (to reduce future casework, we also take $c$ to be even when $k$ is even). As before, the costs associated to intergroup edges will be 1, while the costs of intragroup edges, 0.  Our cost matrix  is $$\hat{C}:=\left(J_{k+1}-I_{k+1}\right)\otimes J_{ck}.$$  

Notice that any integer solution to the $k$-cycle problem will use cycles of length $c(k+1)$, while each group is of size $ck$.  Hence, any cycle in any integer solution will need to use at least two expensive edges.  This lower bounds the cost of $\CYCOPT(\hat{C})$ as $2k$.  We can also achieve this cost by labeling the groups of $ck$ vertices as $G_1,..., G_{k+1}$.  For $i=1, ..., k$, we create a cycle $\mathcal{C}_i$ that visits all vertices in group $G_i$, then visits $c$ vertices in $G_{k+1}$, and then returns to $G_i$ for each $i=1, ..., k$.  Each cycle $\mathcal{C}_i$ costs 2, so that the cost is indeed $2k$.  Hence, regardless of $n$, $\CYCOPT(\hat{C}) = 2k.$

Our proof of Theorem \ref{MainThm2} proceeds  as in the proof of Theorem \ref{MainThm}.  We find solutions whose structure respects the symmetry of $\hat{C}$: solutions that place a weight $a_i$ on each intragroup edge, a weight $b_i$ on each intergroup edge, and zeros on the diagonal. That is,
$$X^{(i)} = \left(\left(b_i J_{k+1}+(a_i-b_i)I_{k+1}\right)\otimes J_{ck}\right)-a_i I_n.$$    The structure of these matrices (and their linear combinations) allows us to again explicitly write down the eigenvalues.
We enforce the same constraints on the row sums of the $X^{(i)},$ which now become $$ b_i  = \f{1}{ck^2} \begin{cases} \left(2-(ck-1)a_i\right), &\text{ if } i<d \\ \left(1-(ck-1)a_i\right), &\text{ if } i=d. \end{cases}$$ 
To prove the theorem, it suffices to show that the following choice of $a_i$, for $i=1, ..., d$, leads to a feasible SDP solution:
 $$a_i = \begin{cases} 0, & \text{ if } i \not\equiv_k k  \\ \f{2}{n-k-1} \left(\cos\left(\f{\pi i}{d}\right)+k\right), &  \text{ if } i  \equiv_k k, i\neq d \\ \f{1}{n-k-1} \left(\cos\left(\f{\pi i}{d}\right)+k\right), & \text{ if } i=d. \end{cases}$$ 
 This is sufficient, as it will imply 
 \begin{align*}
b_k  
\leq \f{\pi^2}{cd^2(k+1)}.
\end{align*}
For these solutions, it can then be shown that
\begin{align*}
\f{\SDPOPT(\hat{C})}{\CYCOPT(\hat{C})} 
&\leq \f{\f{c^2}{2}(k+1)k^3 b_k}{2k}\\
& \leq \pi^2 \f{k}{k+1} \f{1}{n}.
\end{align*}
The full proof of Theorem \ref{MainThm2} uses almost exactly the same ideas as in the proof of Theorem \ref{MainThm}: for the structured solutions, finding a feasible SDP solution is equivalent to finding feasible solutions to a linear program; the $a_i$ indicated above are feasible solutions to this linear program.  The theorem then follows by bounds on $b_k$.    Because the ideas are so similar, we defer a sketch of the details to Appendix \ref{sec:app2}.

\section{Conclusion and Open Questions}
In this paper, we have shown that an  SDP for the TSP introduced in de Klerk et al. \cite{Klerk08} has an unbounded integrality gap.  We then deduced several corollaries, and used the same techniques to show that a related SDP, for the $k$-cycle cover problem and introduced in de Klerk et al. \cite{Klerk12}, also has an unbounded integrality gap.

One open question relates to the relationship between the SDP of de Klerk et al. \cite{Klerk08} and the subtour LP.  For the instance we constructed, the subtour LP outputs the exact cost of a solution tour.  Is it the case that an approximation algorithm that runs both the SDP and subtour LP, then takes the best solution, has an integrality gap of $1.5-\eps$ for $\eps>0?$

A second open question relates to the performance of the SDP on special types of TSP instances.  One example is the Euclidean TSP, where each city $i\in[n]$ corresponds to a point $x_i\in\R^2$, and the cost $c_{ij}$ is given by the Euclidean distance between $x_i$ and $x_j$. While no algorithm for the general TSP (with metric and symmetric edge costs) has been shown to have an integrality gap strictly less than $1.5,$ Arora \cite{Aro96} and Mitchell \cite{Mitch99} give a polynomial time approximation scheme for the Euclidean TSP.    Moreover, one can solve the Euclidean TSP in $\R^1$ exactly: if $x_m \in \min_{i\in [n]}x_i$ and $x_M \in \max_{i\in [n]}x_i,$ then any optimal tour will cost $2(x_M-x_m);$ such a tour can be achieved by starting at $x_m$, iteratively visiting vertices in increasing order of $x_i$ until reaching $x_M$, and returning to $x_m$.  We noted that our instance corresponds to an instance of Euclidean TSP in $\R^1$ (and hence in $\R^2$ by lifting the points $(0)$ and $(1)$ in $\R^1$ to $(0, 0)^T$ and $(1, 0)^T$, respectively), so that the SDP has an unbounded integrality gap even when restricted to the Euclidean TSP (or the Euclidean TSP in $\R^1$).  

Another class of instances that has received considerable attention is that of graphic TSP: here the input corresponds to a connected, undirected graph $G$ on vertex set $[n]$, and for $i, j\in [n],$ the cost $c_{ij}$ is the length of the shortest $i$-$j$ path in $G$. Several recent papers have bounded the integrality gap on graphic TSP instances as strictly less than 1.5.  (See, for example,  Gharan, Saberi, and Singh \cite{Gha11} for a $1.5-\epsilon$ bound, M{\"o}mke and Svensson \cite{Mom16} for a 1.461 bound, Mucha \cite{Muc14} for a $\f{13}{9}\approx 1.444$ bound, and Seb{\H{o}} and Vygen \cite{Seb14} for a $1.4$ bound.)    It is not hard to show that the SDP has at most an integrality gap of 2 when restricted to graphic TSP instances.  This follows because the minimum-cost Hamiltonian cycle is at most twice the cost of a MST (see Section 2.4 of Williamson and Shmoys \cite{DDBook}, for example), and the cost of an MST in a connected graph with unit edge weights is $n-1$.  Conversely, in graphic TSP the minimum-cost of an edge is $1$.  This means that $C\geq J-I$ (entrywise).  Thus:
$$\f{1}{2}\langle C, X^{(1)}\rangle  \geq  \f{1}{2} \langle J-I, X^{(1)}\rangle  = \f{1}{2}\langle J, X^{(1)}\rangle = n.$$
 Above, the first inequality follows from the fact that $X^{(1)}$ is nonnegative, the first equality follows from the fact that the diagonal of $X^{(1)}$ is zero, and the final equality follows from $X^{(1)}e=2e$. 
 An open question is to exactly compute the integrality gap of the SDP on graphic TSP; we conjecture that the integrality gap is at least $1.5$ and is asymptotically achieved when $G$ is a path.
\clearpage

\bibliography{bibliog} 
\bibliographystyle{abbrv}

\clearpage
\appendix

\section{Omitted Proofs}
First we provide an alternative, direct proof of the following theorem from de Klerk et al.\ \cite{Klerk08}, and then prove a result stated in Section \ref{sec:LPSDP}.
\begin{thm}\label{thm:apthm}
Let $n$ be even and let $X^{(1)}, ..., X^{(d)}$ be feasible for the SDP (\ref{eq:SDP}).  Then $$X^{(i)}e=\begin{cases} 2e, & \text{ if } i<d \\ e, & \text{ if } i=d.\end{cases}$$
\end{thm}
This proof has three main steps.  First, we let $Q$ be the matrix of coefficients of the $X^{(i)}$ in the PSD constraints from the SDP (\ref{eq:SDP}).  We compute $Q^{-1}$ and observe that all entries have the same sign.  Hence, for our second step, we can take certain nonnegative linear combinations of the PSD constraints to find (weaker) PSD constraints on each individual $X^{(i)}.$   Finally, we show that these PSD constraints imply the theorem statement.  The proof is not hard, but it is notationally complex, so we first illustrate the main ideas with an example.

\begin{exam} When $n=6$, the PSD constraints are
$$\begin{cases}
\f{1}{2}X^{(1)}-\f{1}{2}X^{(2)}-X^{(3)} \succeq -I & \hspace{5mm} (i)\\
-\f{1}{2}X^{(1)} -\f{1}{2}X^{(2)} +X^{(3)}\succeq -I & \hspace{5mm}(ii)\\
-X^{(1)}+X^{(2)}-X^{(3)}\succeq -I &\hspace{5mm} (ii).
\end{cases}$$
So that the matrix of coefficients and its inverse are:
$$Q=\begin{pmatrix} \f{1}{2} & -\f{1}{2} & -1 \\
-\f{1}{2} & -\f{1}{2} & 1 \\
-1 & 1 & -1 \end{pmatrix}, \hspace{5mm} Q^{-1}= \begin{pmatrix} -\f{1}{3} & -1 & -\f{2}{3} \\ -1 & -1 & 0 \\ -\f{2}{3} & 0 & -\f{1}{3} \end{pmatrix}.$$
Now consider the first row.  All of its entries are nonpositive, and so we take a nonnegative linear combination of the PSD constraints dictated by the first row of $-Q^{-1}.$  Here, we take $\f{1}{3}(i) + (ii) + \f{2}{3} (iii)$, which yields an inequality on just $X^{(1)}$: $$2I \succeq X^{(1)}.$$  Obtaining similar equations for each $X^{(i)}$ will let us deduce the stated theorem.
\end{exam}

We begin with the following, which contains the first two steps of the proof.
\begin{cm}\label{cm:acm}
Let $X^{(1)}, ..., X^{(d)}$ be feasible for the SDP (\ref{eq:SDP}).  Then $$X^{(i)}\preceq \begin{cases} 2I, &\text{ if }  i<d \\ I, &\text{ if }  i=d.\end{cases}$$
\end{cm}
\begin{proof}[Proof (of Claim)]
Let $$Q=\left(\cos\left(\f{2\pi ij}{n}\right)\right)_{i, j=1}^d$$ be the matrix of coefficients of the  $X^{(i)}$ in the PSD constraints of the SDP (\ref{eq:SDP}).  As before, $e\in \R^d$ is the vector of all ones.  Also let  $e_d$ denote the unit vector in $\R^d$ with a one in the $d$-th coordinate, and let $q_d = Qe_d$ be the last column of $Q$.  Note that, by symmetry of $Q$, this has the same entries as the first row of $Q$, $e_d^TQ$, and for $n$ even, $(q_d)_i=\cos(\f{2\pi i d}{d}) = (-1)^i.$  We claim that 
$$Q^{-1} = \begin{cases}
\f{1}{d} \left(2Q - 2ee^T + e_de^T  + ee_d^T - e_dq_d^T -q_de_d^T\right), & \text{ if } d \text{ even }\\
\f{1}{d} \left(2Q - 2ee^T + e_de^T  + ee_d^T-e_dq_d^T -q_de_d^T - e_de_d^T\right), & \text{ if } d \text{ odd. }\\
\end{cases}$$

Expanding entrywise when  $d$ is even:
$$(Q^{-1})_{ij}=\begin{cases}
\f{1}{d}\left[2\cos(\f{2\pi ij}{n})-2\right], & \text{ if } i, j\neq d\\
\f{1}{d}\left[(-1)^i-1\right], & \text{ if } i\neq d, j=d\\
\f{1}{d}\left[(-1)^j-1\right], & \text{ if } i= d, j\neq d\\
0, & \text{ if } i=j= d.
\end{cases}$$  When $d$ is odd, all that changes is that $Q^{-1}_{dd} = -\f{1}{d}.$ 
 This inverse  can be verified by direct computation and applying the product-to-sum identity for cosines with Lemma \ref{lem:trig}.  We do so for $d$ even below in Claim \ref{cm:exs}.

Notice that all entries of $Q^{-1}$ are nonpositive (since $\cos(x)\leq 1$) and that $Q^{-1}e = -2e+e_d$ (we verify this fact in Claim \ref{cm:exs} below when $d$ is even).  That is, each of the first $d-1$ rows sum to $-2$ and the last row sums to $-1$.

We use these equations to take nonnegative linear combinations of the linear matrix inequalities arising as constraints in the SDP, when arranged in the form
\begin{equation}\label{eq:SDPcons} \sum_{s=1}^d \cos\left(\f{2\pi sr}{n}\right)X^{(s)} \succeq -I, \hspace{5mm} r=1, ..., d.\end{equation}
In particular, let $A_{i:}$ denote the $i$-th row of a matrix $A$ (as a row vector) and let $A_{:j}$ denote the $j$-th column of matrix $A$ (as a column vector). Then we write $Q^{-1}Q=I$ as $$(Q^{-1})_{i:}Q_{:j} = \begin{cases} 1, & \text{ if } i=j \\ 0, & \text{ otherwise.}\end{cases}$$  Since $Q$ contains the coefficients of the $X^{(i)}$ (and since each row of $Q^{-1}$ has entries all of the same sign) we can use this relationship to isolate a positive semidefinite constraint for each $X^{(i)}$. To write a constraint on $X^{(i)},$ we take the linear combination of the linear matrix inequalities in Equation (\ref{eq:SDPcons}) dictated by the $i$-th row of $-Q^{-1}$: we take $-(Q^{-1})_{ir}$ times the $r$-th linear matrix inequality.  Doing so yields a coefficient of on $X^{(j)}$ of exactly $$-(Q^{-1})_{i:}Q_{:j} = \begin{cases} -1, & \text{ if } i=j \\ 0, & \text{ otherwise.}\end{cases}$$  The coefficient on the right side (on $-I$) is $$-(Q^{-1})_{i:}e = \begin{cases} 2, & \text{ if }  i<d \\ 1, & \text{ if } i=d.\end{cases}$$  Summarizing, we obtain the linear matrix inequalities $$-X^{(i)}\succeq \begin{cases} -2I, &\text{ if }  i<d \\ -I, &\text{ if }  i=d.\end{cases}$$  \hfill
\end{proof}

We now verify two of our stated computations from Claim \ref{cm:acm}.

\begin{cm}\label{cm:exs}
For $Q=\left(\cos\left(\f{2\pi ij}{n}\right)\right)_{i, j=1}^d$ and $n, d$ even, $$(Q^{-1})_{ij}=\begin{cases}
\f{1}{d}\left[2\cos(\f{2\pi ij}{n})-2\right], & \text{ if } i, j\neq d\\
\f{1}{d}\left[(-1)^i-1\right], & \text{ if } i\neq d, j=d\\
\f{1}{d}\left[(-1)^j-1\right], & \text{ if } i= d, j\neq d\\
0, & \text{ if } i=j= d.
\end{cases}$$   Moreover, $Q^{-1}e = -2e+e_d$
\end{cm}

\begin{proof}
We first verify that that the formula given for $(Q^{-1})_{ij}.$  If $j<d,$ then $(QQ^{-1})_{ij} = \sum_{k=1}^d Q_{ik}Q^{-1}_{kj}$ expands as:
 \begin{align*}
 (QQ^{-1})_{ij}&=\f{1}{d}\left(\left(\sum_{k=1}^{d-1} \cos\left(\f{2\pi ik}{n}\right)  \left(2\cos\left(\f{2\pi jk}{n}\right) -2\right) \right)+ \cos\left(\f{2\pi id}{n}\right)\left((-1)^j-1\right)\right)\\
 &=\f{1}{d}\left(\left(\sum_{k=1}^{d} \cos\left(\f{2\pi ik}{n}\right)  \left(2\cos\left(\f{2\pi jk}{n}\right) -2\right) \right)- \cos\left(\f{2\pi id}{n}\right)\left((-1)^j-1\right)\right) \\
  &=\f{1}{d}\left(\left(\sum_{k=1}^{d}\cos\left(\f{2\pi (i+j)k}{n} \right) +\cos\left(\f{2\pi (i-j)k}{n} \right)  - 2\cos\left(\f{2\pi ik}{n} \right)  \right)- \cos\left(\f{2\pi id}{n}\right)\left((-1)^j-1\right)\right). 
  \intertext{Using  Lemma \ref{lem:trig}, and considering separately the special case when $i=j$:}
    &=\begin{cases} \f{1}{d}\left(\f{(-1)+(-1)^{i+j}}{2} +\f{(-1)+(-1)^{i-j}}{2} -2 \f{(-1)+(-1)^{i}}{2}  - (-1)^i((-1)^j-1)\right), &i\neq j \\
   \f{1}{d}\left(\f{(-1)+(-1)^{i+j}}{2} +d -2 \f{(-1)+(-1)^{i}}{2}  - (-1)^i((-1)^j-1)\right), &i=j \end{cases}
     \intertext{In the first case, we use the fact that $i+j$ and $i-j$ have the same parity.  In the second, we use that $i=j$ implies $i+j$ is even:}
    &=\begin{cases} \f{1}{d}\left((-1)+(-1)^{i+j} -\left((-1)+(-1)^{i}\right)  - (-1)^i((-1)^j-1)\right), &i\neq j \\
   \f{1}{d}\left(d - \left((-1)+(-1)^{i}\right)  - (-1)^i((-1)^j-1)\right), &i=j \end{cases} \\
       &=\begin{cases} \f{1}{d}\left( -1 + (-1)^{i+j} + 1 - (-1)^i - (-1)^{i+j} +(-1)^i \right), &i\neq j \\
   \f{1}{d}\left(d +1 - (-1)^i -1 + (-1)^i \right), &i=j \end{cases} \\
          &=\begin{cases} 0,  &i\neq j \\
  1, &i=j. \end{cases}
 \end{align*}
 When $j=d$ and $d$ is even, we have
  \begin{align*}
 (QQ^{-1})_{id}&=\f{1}{d}\left(\left(\sum_{k=1}^{d-1} \cos\left(\f{2\pi ik}{n}\right) \left((-1)^k -1\right) \right)+ \cos\left(\f{2\pi id}{n}\right)\cdot 0\right)\\
 &=\f{1}{d}\left(\left(\sum_{k=1}^{d-1} \cos\left(\f{2\pi ik}{n}\right) \left(\cos\left(\f{2\pi kd}{n}\right) -1\right) \right) \right).
 \intertext{When $k=d$, the summand is zero since $d$ is even.}
  &=\f{1}{d}\left(\sum_{k=1}^{d} \cos\left(\f{2\pi ik}{n}\right) \left(\cos\left(\f{2\pi kd}{n}\right) -1\right) \right) \\
    &=\f{1}{d}\sum_{k=1}^{d} \left( \cos\left(\f{2\pi ik}{n}\right)\cos\left(\f{2\pi kd}{n}\right) -\cos\left(\f{2\pi ik}{n}\right) \right)\\
        &=\f{1}{d}\sum_{k=1}^{d} \left( \f{\cos\left(\f{2\pi(d+i)k}{n}\right)+\cos\left(\f{2\pi k(d-i)}{n}\right)}{2} -\cos\left(\f{2\pi ik}{n}\right) \right)\\
 &=\begin{cases}
                \f{1}{d}\left( \f{-1+(-1)^{d-i}-1+(-1)^{d+i}}{4} -\f{-1+(-1)^i}{2} \right), i\neq d \\
            \f{1}{d} \left( \f{d+d}{2} -\f{-1+(-1)^i}{2}\right), i= d.      
                \end{cases}
                \intertext{In the former case, we note that $d$ being even implies $i, i+d,$ and $d-i$ all have the same parity.  In the latter case, we note that $i=d$ means $i$ is even:}
 &=\begin{cases}
              0, i\neq d \\
           1, i= d.      
                \end{cases}
 \end{align*}

We now verify  $Q^{-1}e = -2e+e_d$ for $d$ even.  Here:
\begin{align*}
Q^{-1}e &=
\f{1}{d} \left(2Q - 2ee^T + e_de^T  + ee_d^T - e_dq_d^T -q_de_d^T\right)e\\ 
&=\f{1}{d} \left(2Qe - 2ee^Te + e_de^Te  + ee_d^Te - e_dq_d^Te -q_de_d^Te\right)\\
&=\f{1}{d} \left(2Qe - 2de + d e_d  + e - \left(\sum_{i=1}^d (-1)^i \right) e_d-q_d\right).
\intertext{The $i$-th entry of $Qe$ is $\sum_{k=1}^d \cos\left(\f{2\pi i k}{n}\right).$  By  Lemma \ref{lem:trig}, this is $\f{-1+(-1)^i}{2}.$  Since $(q_d)_i=(-1)^i$, we have that $2Qe-q_d=-e$.  Also, since $d$ is even, $ \left(\sum_{i=1}^d (-1)^i \right)=0.$  Hence:}
&=\f{1}{d} \left(-e - 2de + d e_d  + e \right)\\
&= -2e+e_d.
\end{align*}  \hfill
\end{proof}

Using these linear matrix inequalities, we can now prove Theorem \ref{thm:apthm}.
\begin{proof}[Proof (of Theorem \ref{thm:apthm})]

By Claim \ref{cm:acm}, for any $v\in \R^n,$ we have $v^TX^{(i)}v\leq 2v^Tv$ for $i<d$, and $v^TX^{(d)}v\leq v^Tv.$  
These inequalities imply
\begin{align*}
n^2-n = e^T(J-I)e &= e^T\left(\sum_{i=1}^d X^{(d)} \right)e \\
&\leq \left(\sum_{i=1}^{d-1} e^T(2I)e \right)+ e^TIe\\
&= 2(d-1)n + n \\&= n^2-n.
\end{align*}
Equality must prevail throughout, so that $e^TX^{(i)}e = 2n$ for $i<d$, and $e^TX^{(d)}e=n.$

Now suppose that $i<d$ and $X^{(i)}e \neq 2e.$ The fact that $e^TX^{(i)}e=2n$ implies that there is a row whose sum is strictly greater than 2: say $2+2\eps$ with $\eps>0$.  Without loss of generality, we take this to be row 1, and write the row sum as $\sum_{j=1}^n X^{(i)}_{1j}=2+2\eps$ for $\eps>0$.  We obtain a contradiction to our positive semidefinite relationships by computing $(e+\eps e_1)^TX^{(i)}(e+\eps e_1),$ where $e_1\in\R^d$ is the unit vector with a 1 in the first coordinate.
\begin{align*}
(e+\eps e_1)^T X^{(i)} (e+\eps e_1) &= e^TX^{(i)}e + 2\eps e_1^T X^{(i)} e + \eps^2 e_1^T X^{(i)}e_1 \\
&= 2n+2\eps (2+2\eps) + 0 \\
&= 2n + 4\eps + 4\eps^2  \\
& > 2n + 4\eps + 2\eps^2 \\
&= 2(e+\eps e_1)^TI (e+\eps e_1).
\end{align*}
With $X^{(d)}$ we obtain an analogous contradiction by supposing that $\sum_{j=1}^n X^{(i)}_{1j}=1+\eps$. \hfill
\end{proof}

The above  result extends when $n$ is odd: the SDP of de Klerk et al.\ \cite{Klerk08} sets $d=\lfloor\f{n}{2}\rfloor$ and is otherwise  the  same:
$$
\begin{array}{l l l}
\min & \f{1}{2} \tr\left(CX^{(1)}\right) & \\
\text{subject to} & X^{(k)} \geq 0, & k=1, \ldots, d \\
& \sum_{j=1}^d X^{(j)} = J-I, & \\
& I + \sum_{j=1}^d \cos \left(\f{2\pi jk}{n}\right) X^{(j)} \succeq 0, & k=1, \ldots, d \\
& X^{(k)} \in S^n, & k=1, \ldots, d.
\end{array} $$
When $n$ is odd, all of the $d$ distance matrices have row sums equal to 2, so we now expect $X^{(i)}e=2e$ for $i=1, ..., d$; indeed this is the case.  Our proof above generalizes, except $Q^{-1}$ takes a simpler form:  $$(Q^{-1})_{ij} = \f{4}{n}\left(\cos\left(\f{2\pi ij}{n}\right) - 1\right).$$  In this case we obtain that $$2I-X^{(i)} \succeq 0, \hspace{5mm} i=1, ..., d,$$ from which $X^{(i)}e=2e$ follows.

This machinery lets us prove a result stated informally in Section \ref{sec:LPSDP}.  Let $$X^{(1)}_P = \{(X^{(1)}_{ij})_{1\leq i < j \leq n}, X^{(1)} \text{ feasible for the SDP}\}\subset \R^{\ch{5}{2}}$$ be the projection of feasible  $X^{(1)}$ for the SDP (\ref{eq:SDP}) when $n=5$ onto $\R^{\ch{5}{2}}.$ 
\begin{prop}\label{prop:5case}
$X^{(1)}_P $ is exactly equal to the set of feasible solutions for the subtour LP when $n=5$.  Moreover, both are equal to the set of convex combinations of Hamiltonian cycles.
\end{prop}

Our proof follows from a lemma that lets us write the feasible solution for the subtour LP without the subtour elimination constraints.  Moreover, this lemma says, for $n=5$, the degree constraints (and subtour LP) perfectly capture the convex hull of all feasible Hamiltonian cycles.  For $F\subset E$ let $\chi^F \in \R^{\ch{n}{2}}$ be the incidence vector of $F$: $$\chi^F_e = \begin{cases} 1, & e\in F \\ 0, & \text{ else.} \end{cases}$$  Also, let $\text{conv}\{v_1, ..., v_k\}$ be the convex hull of vectors $v_1, ..., v_k.$  
%
\begin{lm}\label{lem:simpsubtour}
When $n=5$, \normalfont
\begin{align*}
\{&x \in \R^{\ch{5}{2}}:  \sum_{e\in \delta(\{v\})} x_e = 2 \fa v\in [5], 0\leq x_e \leq 1 \fa e\in E\}\\
& =\{x \in \R^{\ch{5}{2}}:  \sum_{e\in \delta(\{v\})} x_e = 2 \fa v\in [5], 0\leq x_e \leq 1 \fa e\in E, \sum_{e\in \delta(S)} x_e \geq 2 \fa \emptyset \subset S\subset [5]\}  \\
&= \text{conv}\{ \chi^F: F \text{ is a Hamiltonian cycle} \}.
\end{align*}
\end{lm}

This result is mentioned on page 286 of Gr{\"o}tschel and Padberg \cite{Gro86}; the degree constraints and the constraints that  $0\leq x_e \leq 1$  imply all subtour constraints when $n=5.$  

\begin{proof}[Proof (of Proposition \ref{prop:5case})]
In Theorem \ref{thm:apthm} and its equivalent for $n$ odd, we argued that any feasible $X^{(1)}$ for the SDP  (\ref{eq:SDP}) met the degree constraints.  By Lemma \ref{lem:simpsubtour}, this is sufficient to imply $$X^{(1)}_P \subset \text{conv}\{ \chi^F: F \text{ is a Hamiltonian cycle} \}  .$$   By Proposition \ref{prop:relax}, however,  $$X^{(1)}_P \supset \text{conv}\{ \chi^F: F \text{ is a Hamiltonian cycle} \}  .$$  Hence these two sets are equal, and again by Lemma \ref{lem:simpsubtour}, they are both also equal to the set of feasible solutions to the subtour LP. \hfill
\end{proof}

\section{Details for the $k$-Cycle Cover Problem}\label{sec:app2}
Here we sketch the proof of Theorem \ref{MainThm2}.  We recall that this SDP searches for $k$ equally sized cycles covering all $n=ck(k+1)$ vertices, where $n$ scales with $c\in \N$ (and if $k$ is even, we require $c$ to be even to reduce casework).  The cost matrix is $$\hat{C}:=\left(J_{k+1}-I_{k+1}\right)\otimes J_{ck},$$ and we found that $\CYCOPT(\hat{C}) = 2k.$

\begingroup
\def\thetheorem{\noindent {\bf Theorem \ref{MainThm2}}}
\begin{thetheorem} \emph{$$\SDPOPT(\hat{C}) \leq \f{\pi^2}{n}\f{k}{k+1} \CYCOPT(\hat{C}).$$ }
\end{thetheorem}
\endgroup

%

We recall that our cost matrix  is $$\hat{C}:=\left(J_{k+1}-I_{k+1}\right)\otimes J_{ck}.$$  
We look for solutions of the form
$$X^{(i)} = \left(\left(b_i J_{k+1}+(a_i-b_i)I_{k+1}\right)\otimes J_{ck}\right)-a_i I_n,$$  
where\footnote{As before, these conditions enforce that the row sums are correct.  I.e., that $$(ck-1)a_i + ck^2b_i = \begin{cases} 2, & \text{ if }  i=1, ..., d-1 \\ 1, & \text{ if } i=d. \end{cases}$$} $$ b_i  = \f{1}{ck^2} \begin{cases} \left(2-(ck-1)a_i\right), &\text{ if } i<d \\ \left(1-(ck-1)a_i\right), &\text{ if } i=d. \end{cases}$$ 
We ultimately show that the following is a feasible choice of the $a_i$:
 $$a_i = \begin{cases} 0, & \text{ if } i \not\equiv_k k  \\ \f{2}{n-k-1} \left(\cos\left(\f{\pi i}{d}\right)+k\right), &  \text{ if } i  \equiv_k k, i\neq d \\ \f{1}{n-k-1} \left(\cos\left(\f{\pi i}{d}\right)+k\right), & \text{ if } i=d. \end{cases}$$ 

Also, note that in each row of the $n$ rows of $X^{(k)}$ there are $n-ck=ck^2$ entries that are $b_k$ (occurring exactly where $C$ has 1s), and $ck$ entires that are $a_k$ (occurring exactly where $C$ has 0s).  Hence the cost of such a solution is now based entirely on $b_k$:
$$\f{1}{2} \langle C, X^{(k)}\rangle = \f{1}{2}n(n-ck)b_k = \f{c^2}{2}(k+1)k^3 b_k.$$ 

\subsection{Writing an Equivalent Linear Program}
We first prove the following counterpart of Proposition \ref{prop:SDPLP}.
\begin{prop}
Finding a minimum-cost feasible solution of the form $$X^{(i)} = \left(\left(b_i J_{k+1}+(a_i-b_i)I_{k+1}\right)\otimes J_{ck}\right)-a_i I_n $$ for $i=1, ..., d$ with $$ b_i  = \f{1}{ck^2} \begin{cases} \left(2-(ck-1)a_i\right), &\text{ if } i<d \\ \left(1-(ck-1)a_i\right), &\text{ if } i=d\end{cases}$$ is equivalent  to solving the following optimization problem: \normalfont
\begin{equation}\label{kcycLP}
\begin{array}{llll}
\max & a_k &  \\
\text{subject to} & \sum_{i=1}^d    \cos\left(\f{2\pi ij}{n}\right)  a_i & \geq -\f{1}{ck-1}, & j=1, ..., d\\
 & \sum_{i=1}^d    \cos\left(\f{2\pi ij}{n}\right)  a_i & \leq 1, &  j=1, ..., d\\
& \sum_{i=1}^d a_i &= 1 \\
& a_i & \leq \f{2}{ck-1},  & i=1, ..., d-1\\
& a_d& \leq \f{1}{ck-1}\\
& a_i & \geq 0,  & i=1, ..., d. 
\end{array} \end{equation}
\end{prop}

\begin{proof}[Proof (Sketch)]
 Just as before, the $j$-th SDP constraint for the $k$-cycle cover problem is
$$I + \sum_{i=1}^d \cos\left(\f{2\pi ij}{n}\right) X^{(i)} \succeq 0.$$  We again define 
 $$a^{(j)} =\sum_{i=1}^d    \cos\left(\f{2\pi ij}{n}\right)  a_i, \hspace{5mm} b^{(j)} =\sum_{i=1}^d    \cos\left(\f{2\pi ij}{n}\right)  b_i.$$  Then the $j$-th SDP  constraint becomes
\begin{align*}
I_n + \sum_{i=1}^d \cos\left(\f{2\pi ij}{n}\right) X^{(i)} 
&= I_n + \sum_{i=1}^d \cos\left(\f{2\pi ij}{n}\right) \left( \left(\left(b_i J_{k+1}+(a_i-b_i)I_{k+1}\right)\otimes J_{ck}\right)-a_i I_n\right)\\
&= (1-a^{(j)})I_n +   \left( \left(b^{(j)} J_{k+1}+(a^{(j)}-b^{(j)})I_{k+1}\right)\otimes J_{ck}\right).
\end{align*}

We write down the eigenvalues of the matrix representing the  $j$-th SDP constraint  using properties of Kronecker products and shifts by the identity matrix.  For example, $J_{ck}$ has $0$ as an eigenvalue with multiplicity $ck-1$.  These give rise to $0$ as an eigenvalue of the term $ \left(b^{(j)} J_{k+1}+(a^{(j)}-b^{(j)})I_{k+1}\right)\otimes J_{ck}$ with multiplicity  $(k+1)(ck-1)$.  Accounting for the $(1-a^{(j)})I_n$ term, the eigenvalue $$1-a^{(j)}$$ occurs with multiplicity $(k+1)(ck-1)$.

Proceeding similarly  by casework, we obtain the following eigenvalue with multiplicity $k$: 
 $$1 + (ck-1)a^{(j)} - ckb^{(j)}.$$
We also obtain the following eigenvalue with multiplicity $1$: $$1 + (ck-1)a^{(j)} + ck^2 b^{(j)}.$$

 Hence, for the $j$-th psd constraint to hold, it suffices that the following three linear inequalities to hold: 
\begin{equation}\label{kcyceig} 1 - a^{(j)} \geq 0, \hspace{5mm}  1 + (ck-1)a^{(j)} - ckb^{(j)} \geq 0,  \hspace{5mm}1 + (ck-1)a^{(j)} + ck^2 b^{(j)} \geq 0.\end{equation}

As before, we  simplify these by writing $b^{(j)}$ as a function of $a^{(j)}$.  To do so, we again use Lemma \ref{lem:trig}.  We find
\begin{align*}
b^{(j)} 
&= -\f{1}{ck^2} - \f{ck-1}{ck^2}a^{(j)}.
\end{align*}
Substituting this relationship in, our requirements in Equation (\ref{kcyceig}) become
only that $$-\f{1}{ck-1}\leq a^{(j)} \leq 1.$$

Note also that, for solutions to be feasible, we again require that $a_i \geq 0, b_i\geq 0,$ and $ \sum_{i=1}^{d} a_i = 1.$ The requirement $\sum_{i=1}^d b_i=1$ again follows from the relationships between each $a_i$ and $b_i$ as well as $\sum_{i=1}^{d} a_i = 1$; together these two constraints will imply  $\sum_{i=1}^d X^{(i)} = J-I.$   \hfill
\end{proof}

\subsection{ Finding Feasible Solutions to the Linear Program}
We make one more assumption to reduce our casework: if $k$ is even, we only consider even $c$.  This enforces that $d$ is always congruent to $0$ mod $k$.  With this assumption, we claim that the following is a feasible solution to the linear program  (\ref{kcycLP}):
$$a_i = \begin{cases} 0, & \text{ if } i \not\equiv_k k  \\ \f{2}{n-k-1} \left(\cos\left(\f{\pi i}{d}\right)+k\right), &  \text{ if } i  \equiv_k k, i\neq d \\ \f{1}{n-k-1} \left(\cos\left(\f{\pi i}{d}\right)+k\right), & \text{ if } i=d. \end{cases}$$

Since $\cos(x)\geq - 1$, we observe that the $a_i$ are nonnegative.  Also,  
$$\f{\cos\left(\f{\pi i}{d}\right) + k}{n-k-1}\leq \f{k+1}{n-k-1} = \f{k+1}{ck(k+1)-(k+1)}  = \f{1}{ck-1},$$  so that the $a_i$ satisfy their upper bounds.  Now it remains to show that the $a_i$ sum to 1 and that they satisfy the $a^{(j)}$ constraints. We begin with a modification of Lemma  \ref{lem:trig}.
\begin{lm}\label{lem:trig2}
Let $s\geq 0$ be an integer such that $\f{\pi ks}{d}\notin \{0, 2\pi, 4\pi, ...\}$.  Then 
$$\sum_{i=1}^{d/k} \cos\left(\f{\pi k i}{d}s\right) = \f{-1+(-1)^s}{2}.$$
\end{lm}

\begin{proof}
Again, this is a specific instance of Lagrange's trigonometric identity:
$$\sum_{i=1}^N \cos(i\theta) = -\f{1}{2} + \f{\sin\left(\left(N+\f{1}{2}\right)\theta\right)}{2\sin\left(\f{\theta}{2}\right)}.$$  Here we take $N=\f{d}{k}$ and $\theta = \f{\pi ks}{d}.$  \hfill
\end{proof}

We now consider the sum of the $a_i$.
\begin{cm}
For $$a_i = \begin{cases} 0, & \text{ if } i \not\equiv_k k  \\ \f{2}{n-k-1} \left(\cos\left(\f{\pi i}{d}\right)+k\right), & \text{ if } i  \equiv_k k, i\neq d \\ \f{1}{n-k-1} \left(\cos\left(\f{\pi i}{d}\right)+k\right), & \text{ if } i=d, \end{cases}$$ we have 
$$\sum_{i=1}^d a_i = 1.$$
\end{cm}

\begin{proof}
We use Lemma \ref{lem:trig2} with $s=1$.  In this case $\f{\pi k}{d}$ is not an integer multiple of $2\pi$, so the proof follows as for Claim \ref{cm:asum}.  Skipping most of the algebra, we obtain:
\begin{align*}
\sum_{i=1}^d a_i 
&= \f{1}{n-k-1}\left( \left( \sum_{i=1}^{\f{d}{k}} 2 \left(\cos\left(\f{\pi i k}{d}\right)+k\right)   \right) - \cos\left(\f{\pi d}{d}\right) - k \right)\\
&=1.
\end{align*}  \hfill
\end{proof}

\begin{cm}
With $$a_i = \begin{cases} 0, & \text{ if } i \not\equiv_k k  \\ \f{2}{n-k-1} \left(\cos\left(\f{\pi i}{d}\right)+k\right), & \text{ if } i  \equiv_k k, i\neq d \\ \f{1}{n-k-1} \left(\cos\left(\f{\pi i}{d}\right)+k\right), & \text{ if } i=d \end{cases},$$ we have $$-\f{1}{ck-1}\leq a^{(j)} \leq 1.$$
\end{cm}

Our proof is in the same spirit as that of Lemma \ref{cm:aj}.
\begin{proof}
We begin by writing the terms in one sum to $\f{d}{k}$ to use Lemma \ref{lem:trig2}, and then use the product to sum  identity for cosines.  
\begin{align*}
 a^{(j)} &= \sum_{i=1}^d \cos\left(\f{2\pi ij}{n} \right)a_i \\
&=  \f{ (-1)^{j+1} \left( k-1 \right) + 2 k \sum_{i=1}^{d/k} \cos\left(\f{\pi ijk}{d} \right) +  \sum_{i=1}^{d/k} \cos\left(\f{\pi ki}{d}(j+1) \right)+\sum_{i=1}^{d/k}\cos\left(\f{\pi k i}{d}(j-1)\right)}{n-k-1}\\
&:= (*).
\end{align*}
However, we now need to be  careful with Lagrange's identity, as there are more cases where we cannot apply Lagrange's trigonometric identity to these sums.  In particular, we cannot do so when summing  $\cos\left(\f{\pi k s}{d}\right)$ where  $\f{\pi k s}{d}$ is an integer multiple of $2\pi$.  This happens when $ks =\lambda 2d=\lambda n$ for $\lambda = 0, 1, 2, ....$. In terms of $s$, this is when
$$s = \lambda \f{n}{k} = \lambda c (k+1),$$  and for $k\geq 2$ this can happen more than just when $s=0$.  
Notice, though, that if $s=\lambda c(k+1)$, then $s$ must be even (as either $k+1$ or $c$ is even).  Also, the distance between successive values of $s$ where this occurs is $c(k+1)\geq 3$ for $k\geq 2$.  In simplifying $(*)$ we will need to evaluate sums of $\cos\left(\f{\pi k s}{d}\right)$ when $s=j-1, j$, and $j+1$; at most one of these terms can be an integer multiple of $c(k+1)$.  Summarizing:
$$ \sum_{i=1}^{d/k} \cos\left(\f{\pi k i}{d}s\right) = \begin{cases} \f{d}{k}, & \text{ if } s = \lambda c(k+1) \text{ for } \lambda =0, 1, 2, ... \\ \f{-1+(-1)^s}{2}, & \text{ else.} \end{cases}$$
We now consider three cases: none of the three sums evaluate to $\f{d}{k}$, the sum $ \sum_{i=1}^{d/k} \cos\left(\f{\pi ijk}{d} \right)$ evaluates to $d/k$, and either of $ \sum_{i=1}^{d/k} \cos\left(\f{\pi ik}{d}j\pm 1 \right)$ evaluate to $d/k$.  We respectively obtain the following, noting that $j+1$ and $j-1$ both have the same parity.  Up to bookkeeping, the algebraic manipulation is as in Claim \ref{cm:aj}, so we  summarize the evaluations:

Case 1):
\begin{align*}
(*)
&=- \f{k+1}{n-k-1} = - \f{k+1}{ck(k+1)-(k+1)} = -\f{1}{ck-1}
\end{align*}

Case 2), in which $j$ is even:
\begin{align*}
 (*) &= \f{ - \left( k-1 \right) +2d - 2}{n-k-1}= \f{n-k-1}{n-k-1} =1\\
\end{align*}

Case 3), in which $j\pm 1$ is even.  Without loss of generality, we take the case where $j+1$ evaluates to $d/k$.  Then:
\begin{align*}
(*)&=  \f{ k-1 -2k + \f{d}{k}}{n-k-1}
= \f{ c(k+1)-2(k+1)}{2(n-k-1)}
= \f{c-2}{2(ck-1)}\\
\end{align*}
For all $c\geq 1, k\geq 2$, we find that in every case $$-\f{1}{ck-1} \leq (*) \leq 1,$$ completing the proof.   \hfill
\end{proof}

\subsection{The Unbounded Integrality Gap} 
We are now able to prove the generalization of Theorem \ref{MainThm}:

\begingroup
\def\thetheorem{\noindent {\bf Theorem \ref{MainThm2}}}
\begin{thetheorem} \emph{$$\SDPOPT(\hat{C}) \leq \f{\pi^2}{n}\f{k}{k+1} \CYCOPT(\hat{C}).$$ }
\end{thetheorem}
\endgroup

\begin{proof}
First recall that the optimal cost of the SDP is  $$\SDPOPT(\hat{C})\leq\f{c^2}{2}(k+1)k^3 b_k$$ by the feasible solution we found, and that $$\CYCOPT(\hat{C}) = 2k.$$  
Using  $\cos\left(\f{k\pi}{d}\right) \geq 1-\f{k^2\pi^2}{2d^2}$ we have
\begin{align*}
b_k  
&= \f{2}{ck^2}\left(1-(ck-1)\f{1}{n-k-1} \left(\cos\left(\f{\pi k}{d}\right)+k\right)\right) \\
&\leq \f{2}{ck^2}\left(1-\f{ck-1}{n-k-1} \left(k+1-\f{k^2\pi^2}{2d^2}\right)\right) \\
&= \f{\pi^2}{cd^2(k+1)}.
\end{align*}
For any fixed $k$, the denominator is again $\mathcal{O}(n^3)$.  Hence:
\begin{align*}
\f{\SDPOPT(\hat{C})}{\CYCOPT(\hat{C})} 
&\leq \f{\f{c^2}{2}(k+1)k^3 b_k}{2k}\\
& \leq \f{\pi^2}{4} \f{ck^2}{d^2} \\
\intertext{Using $4d^2=n^2$ and $n=ck(k+1)$}
& = \pi^2 \f{1}{c (k+1)^2 }\\
& = \pi^2 \f{k}{k+1} \f{1}{n}.
\end{align*}
In the last line, we used that $c=n/(k(k+1)).$ Again, the SDP's integrality gap is unbounded; as $n$ increases, solutions to the SDP become arbitrarily small.
 \hfill
\end{proof}

\end{document}